\documentclass{amsart}

\usepackage{ucs}
\usepackage[utf8x]{inputenc}
\usepackage{url}
\usepackage{graphicx,bbm,stmaryrd}
\usepackage{lmodern}
\usepackage
    [pdftex,plainpages=false,linkcolor=blue,citecolor=blue,
    urlcolor=blue,unicode=true,pdfborder={0 0 0}]
    {hyperref}
\usepackage{microtype}
\usepackage{cmap}
\usepackage{amssymb}

\usepackage{tikz}
\usetikzlibrary{arrows}
\usetikzlibrary{calc}
\usetikzlibrary{fit}
\usetikzlibrary{positioning}
\usepackage{booktabs}
\usepackage{multirow}
\usepackage{subfigure}
\newcommand\mytitle{Rigorous Uniform Approximation of D-finite Functions Using Chebyshev Expansions}

\hypersetup{
    pdftitle = {\mytitle},
    pdfauthor = {Alexandre Benoit and Mioara Joldeș and Marc Mezzarobba},
    pdflang = {en}
}

\newtheorem{theorem}{Theorem}[section]
\newtheorem{proposition}[theorem]{Proposition}
\newtheorem{lemma}[theorem]{Lemma}
\newtheorem{corollary}[theorem]{Corollary}
\theoremstyle{definition}
\newtheorem{definition}[theorem]{Definition}
\newtheorem{statealgorithm}[theorem]{Algorithm}
\newtheorem{problem}[theorem]{Problem}
\newtheorem{example}[theorem]{Example}
\newtheorem{question}[theorem]{Question}
\theoremstyle{remark}
\newtheorem{remark}[theorem]{Remark}

\numberwithin{equation}{section}

\newenvironment{algorithm}{%
    \begin{figure}[ht!]
    \hrule
    \begin{statealgorithm}
}{%
    \end{statealgorithm}
    \hrule
    \end{figure}
}

\newenvironment{body}{%
    \begin{list}{\stepcounter{algoline}\tiny{\arabic{algoline}}}{
        \setlength{\topsep}{0pt}%
        \setlength{\itemsep}{0pt}%
        \setlength{\leftmargin}{1em}%
        \setlength{\labelwidth}{1em}%
    }%
}{%
    \end{list}%
}
\newcommand\labelitem[1]{\item\addtocounter{algoline}{-1}%
    \refstepcounter{algoline}\label{#1}}

\newcommand{\Tau}{\mathrm T}
\newcommand\nosymbol{}
\newcommand{\mathd}{\mathrm{d}}

\newcommand{\ninf}[1]{\left\|#1\right\|_{\infty}}

\newcommand{\Cvgseq}{\mathcal C}
\newcommand{\Q}{\mathbbm Q}
\newcommand{\Z}{\mathbbm Z}
\newcommand{\abs}[1]{\left|#1\right|}

\renewcommand{\d}{\mathrm d}

\renewcommand\geqslant\geq
\renewcommand\leqslant\leq
\renewcommand\varepsilon\epsilon
\renewcommand\varphi\phi

\begin{document}

\title[\mytitle]{
  Rigorous Uniform Approximation of D-finite Functions
  Using Chebyshev Expansions
}

\author{Alexandre Benoit}
  \address{%
    Alexandre Benoit;
    Éducation nationale;
    France
  }
  \email{alexandrebenoit@yahoo.fr}
  \urladdr{http://alexandre.benoit.83.free.fr/}
\author{Mioara Joldes}
  \address{%
    Mioara Joldes;
    CNRS, LAAS,
    7 Avenue du Colonel Roche,
    31077 Toulouse, Cedex 4, France 
  }
  \email{joldes@laas.fr}
  \urladdr{http://homepages.laas.fr/mmjoldes/}
\author{Marc Mezzarobba}
  \address{%
    Marc Mezzarobba;
    CNRS, UMR~7606, LIP6, 
         F-75005, Paris, France;
    Sorbonne Universités, UPMC Univ Paris~06, UMR~7606, LIP6,
         F-75005, Paris, France
  }
  \thanks{%
    This research was partly supported by the Austria Science Fund (\emph{FWF}) grants P22748-N18 and Y464-N18,
    and before that by the MSR-Inria joint research center.
  }
  \email{marc@mezzarobba.net}
  \urladdr{http://marc.mezzarobba.net/}


\thanks{This work is in the public domain. As such, it is not
subject to copyright. Where it is not legally possible to consider this
work as released into the public domain, any entity is granted the right
to use this work for any purpose, without any conditions, unless such
conditions are required by law.}

\begin{abstract}

A wide range of numerical methods exists for computing polynomial
approximations of solutions of ordinary differential equations based on
Chebyshev series expansions or Chebyshev interpolation polynomials. We consider
the application of such methods in the context of rigorous computing (where we
need guarantees on the accuracy of the result), and from the complexity point
of view.

It is well-known that the order-$n$ truncation of the Chebyshev expansion of a
function over a given interval is a near-best uniform polynomial approximation
of the function on that interval. In the case of solutions of linear
differential equations with polynomial coefficients, the coefficients of the
expansions obey linear recurrence relations with polynomial coefficients.
Unfortunately, these recurrences do not lend themselves to a direct recursive
computation of the coefficients, owing among other things to a lack of initial
conditions.

We show how they can nevertheless be used, as part of a validated process, to
compute good uniform approximations of D-finite functions together with
rigorous error bounds, and we study the complexity of the resulting algorithms.
Our approach is based on a new view of a classical numerical method going back
to Clenshaw, combined with a functional enclosure method.

\end{abstract}


\maketitle

\tableofcontents

\section{Introduction}
\label{sec:intro}

\subsection{Background}
\label{sec:background}

Many of the special functions commonly used in areas such as mathematical
physics are so-called \emph{D-finite functions}, that is, solutions of linear
ordinary differential equations (LODE) with polynomial
coefficients~{\cite{Stanley1980}}.
This property allows for a uniform theoretic and algorithmic treatment of these functions, an idea that was recognized long ago in Numerical Analysis~\cite[p.~464]{Lanczos1956}, and more recently found many applications in the context of Symbolic Computation~\cite{Zeilberger1990,Salvy2005,Kauers2013}.
The present article is devoted to the following problem.

\begin{problem}
  \label{problem}
  Let~$y:[-1,1] \to \mathbb R$ be a D-finite function specified by a linear differential equation with polynomial coefficients and initial conditions.
  Let~$d \in \mathbbm N$.
  Given~$y$ and~$d$, find the coefficients of a polynomial $p(x) =
  \sum_{n=0}^d c_n T_n (x)$ written on the Chebyshev basis~$(T_n)$, together
  with a ``small'' bound~$B$ such that $| y (x) - p (x) | \leqslant B$ for all
  $x \in [- 1, 1]$.
\end{problem}

Approximations over other real or complex segments (written on the Chebyshev basis adapted to the segment) are reduced to approximations on~$[-1,1]$ by means of an affine change of variables, which preserves D-finiteness.

A first motivation for studying this problem comes from repeated evaluations.
Computations with mathematical functions often require the ability to evaluate a
given function~$y$ at many points lying on an interval, usually with moderate
precision. Examples include plotting, numerical integration, and interpolation.
A standard approach to
address this need resorts to polynomial approximations of~$y$. We deem it
useful to support working with arbitrary D-finite functions in a computer
algebra system. Hence, it makes sense to ask for good uniform polynomial
approximations of these functions on intervals. Rigorous error bounds are
necessary in order for the whole computation to yield a rigorous result.

Besides easy numerical evaluation, polynomial approximations provide a convenient representation of continuous functions on which comprehensive arithmetics including addition, multiplication, composition and integration may be defined.
Compared to the exact representation of D-finite functions by differential equations, the representation by polynomial approximations is only approximate, but it applies to a wider class of functions and operations on these functions.
When we are working over an interval, it is natural for a variety of reasons to write the polynomials on the Chebyshev basis rather than the monomial basis.
In particular, the truncations that occur during most arithmetic operations then maintain good uniform approximation properties.
Trefethen~\emph{et al.}'s Chebfun~\cite{Trefethen2007,DriscollBornemannTrefethen2008} is a popular numerical computation system based on this idea.

In a yet more general setting, Epstein, Miranker and Rivlin developed a
so-called ``ultra-arithmetic'' for functions that parallels floating-point
arithmetic for real numbers~\cite{EpsteinMirankerRivlin1982,
EpsteinMirankerRivlin1982a, KaucherMiranker1984}. Various \emph{generalized
Fourier series}, including Chebyshev series, play the role of floating-point
numbers. Ultra-arithmetic also comprises a function space counterpart of
interval arithmetic, based on truncated series with interval coefficients and
rigorous remainder bounds. This line of approach was revived with the
introduction of ``ChebModels'' in recent work by Brisebarre and
Joldeș~{\cite{BrisebarreJoldes2010}}. Part of the motivation for
Problem~\ref{problem} is to allow one to use arbitrary D-finite functions as
``base functions'' at the leaves of expression trees to be evaluated
using ChebModels.

Finally, perhaps the main appeal of ultra-arithmetic and related techniques is
the ability to solve functional equations rigorously using enclosure
methods~\cite{Moore1979,KaucherMiranker1984,MakinoBerz2003,Neumaier2003,Tucker2011}. LODE
with polynomial coefficients are among the simplest equations to which these
tools apply. A third goal of this article is to begin the study of the
\emph{complexity} of validated enclosure methods, from a computer algebra point of
view, using this simple family of problems as a prototype.

\subsection{Setting}
\label{sec:setting}

To specify the D-finite function~$y$,
we fix a linear homogeneous differential equation of order~$r$ with polynomial
coefficients
\begin{equation}
  \label{eq:deq}
  L \cdot y = a_r y^{(r)} + a_{r - 1} y^{(r - 1)} +
  \cdots + a_0 y = 0, \quad a_i \in \mathbbm{Q} [x] .
\end{equation}
We also write $L = a_r \partial^r + \cdots + a_1 \partial + a_0$.
\label{def:partial}
Up to a change of variable, we assume that we are seeking a polynomial approximation of
a solution~$y$ of~\eqref{eq:deq} over the interval $[-1, 1]$. The uniform norm
on this interval is denoted by $\|\cdot\|_{\infty}$. We also assume that $a_r(x)
\neq 0$ for $x \in [- 1, 1]$, so that (by Cauchy's existence theorem for
complex LODE) all solutions of~\eqref{eq:deq} are
analytic on $[- 1, 1]$.
Besides the operator~$L$, we are given~$r$ initial values
\begin{equation} \label{eq:inicond}
  y^{(i)}(0) = \ell_i, \qquad 0 \leq i \leq r-1.
\end{equation}
Many of the results actually extend to the case of boundary conditions, since we can compute a whole basis of solutions of~\eqref{eq:deq} and reduce boundary value problems to initial value problems by linear algebra.
Also note that the case of initial values given outside the
domain of expansion may be reduced to our setting using numerical analytic
continuation.

\begin{table}
  \begin{center}
    \begin{tabular}{rll}
      \toprule
      $a_i$ & coefficients of the operator~$L$; $a_r(0) \neq 0$ &  p.~\pageref{eq:deq} \\
      $b_i$ & coefficients of the operator~$P$ &  p.~\pageref{item:P} \\
      $\Cvgseq$ & sequences in $\mathbb C^{\mathbb Z}$ with exponential decrease & p.~\pageref{def:Cvgseq} \\
            $\partial$ & differentiation operator, $\partial = \mathrm d/\mathrm d x$ & p.~\pageref{def:partial} \\
      $L$ & differential operator, $L \cdot y = 0$ & p.~\pageref{eq:deq} \\
      $\ell_i$ & initial values, $y^{(i)}(0) = \ell_i$ & p.~\pageref{eq:inicond} \\
      $P$ & Chebyshev recurrence operator & p.~\pageref{thm:rec}, p.~\pageref{def:chebrec} \\
      $p_d^{\ast}$ & degree-$d$ minimax polynomial approximation of~$y$ & p.~\pageref{def:minimax} \\
      $\pi_d$ & truncated Chebyshev expansion operator & p.~\pageref{def:proj} \\
      $s$ & order of~$L$ & p.~\pageref{eq:deq} \\
      $S$ & shift operator, $S: (u_n) \mapsto (u_{n+1})$ & p.~\pageref{def:shift} \\
      $s$ & (usually) half-order of~$P$ & p.~\pageref{item:P}, p.~\pageref{eq:BTrec} \\
      $\mathbf S$ & singularities of~$P$, shifted by~$s$ & p.~\pageref{def:sing} \\
      $T_n$ & Chebyshev polynomials of the first kind & p.~\pageref{sec:chebseries} \\
      $y$ & (usually) unknown function, $L \cdot y = 0$  & p.~\pageref{problem}, p.~\pageref{def:y rat} \\
      $y^{(N)}$ & approximation of~$y$ computed by Algorithm~\ref{algo:Clenshaw} & p.~\pageref{def:yN} \\
      $\hat f$ & inverse Joukowski transform of a function~$f$ & p.~\pageref{def:Joukowski} \\
      $\llbracket i, j \rrbracket$ & integer interval, $\llbracket i, j \rrbracket = \{ i, {i+1}, \dots, j \}$ & p.~\pageref{def:intiv} \\
      \bottomrule
    \end{tabular}
  \end{center}
  \caption{Notation.}
  \label{tab:notation}
\end{table}

Table~\ref{tab:notation} summarizes for quick reference the notation used throughout this article.
Notations related to Chebyshev expansions are detailed in Section~\ref{sec:chebseries} below.
Notations from Theorem~\ref{thm:rec} are also repeatedly used in the subsequent discussion.
Double brackets denote integer intervals $\llbracket i, j \rrbracket = \{ i, {i+1}, \dots, j \}$.
\label{def:intiv}

Unless otherwise noted, we assume for simplicity that all computations are
carried out in exact (rational) arithmetic. The rigor of the computation is
unaffected if exact arithmetic is replaced by floating-point arithmetic in
Algorithm~\ref{algo:Clenshaw} and by interval arithmetic in
Algorithm~\ref{algo:Validation}. (In the case of
Algorithm~\ref{algo:RationalExpansion}, switching to interval arithmetic
requires some adjustments.) However, we do not analyze the effect of rounding
errors on the quality of the approximation polynomial~$p$ and error bound~$B$
from Problem~\ref{problem} when the computations are done in floating-point
arithmetic. In simple cases at least, we expect that
Algorithm~\ref{algo:Clenshaw} exhibits comparable stability to similar methods
based on backward recurrence~\cite{Wimp1984}. Our experiments show a satisfactory numerical behaviour.

To account for this variability in the underlying arithmetic, we assess the
complexity of the algorithms in the \emph{arithmetic model}. In other words, we
only count field operations in~$\mathbbm Q$, while neglecting both the size of
their operands and the cost of accessory control operations.
The choice of the arithmetic complexity model for an algorithm involving
multiple precision numerical computations may come as surprising. Observe however
that all arithmetic operations on both rational and floating-point
numbers of bit size bounded by~$n$ may be done in time
$O(n (\ln n)^{O(1)})$ (see for instance Brent and Zimmermann's book~\cite{BrentZimmermann2010} for detail).
In general, the maximum bit size of the numbers we manipulate is roughly the same as that of
the coefficients of~$p$, which may be checked to be $O(d \ln d)$ when
represented as rational numbers, so that the bit complexity of the algorithm is
actually almost linear in the total \emph{bit} size of the output.

\subsection{Summary of Results}

As we will see in the next section, truncated Chebyshev expansions of analytic functions provide very good approximations of these functions over straight line segments.
In the case of D-finite functions, their coefficients are known to satisfy linear recurrences.
But computing Chebyshev series based on these recurrences is not entirely straightforward.

Roughly speaking, the conclusion of the present article is that these recurrences can nevertheless be used to solve Problem~\ref{problem} efficiently for arbitrary D-finite functions.
The techniques we use (backward recurrence and enclosure of solutions of fixed-points equations in function spaces) date back to the 1950s--1960s.
The originality of this work is that we insist on providing \emph{algorithms} that apply to a well-defined class of functions (as opposed to methods to be adapted to each specific example), and focus on controlling the \emph{computational complexity} of these algorithms.

Our algorithm proceeds in two stages.
We first compute a candidate approximation polynomial, based on the Chebyshev expansion of the function~$y$.
No attempt is made to control the errors rigorously at this point.
We then validate the output using an enclosure method.

The main results of this article are Theorems~\ref{thm:convergence} (p.~\pageref{thm:convergence}) and~\ref{thm:Valid} (p.~\pageref{thm:Valid}), stating respectively that each of these two steps can be performed in linear arithmetic complexity with respect to natural parameters, and estimating the quality of the results they return.
Theorem~\ref{thm:convergence} is based on a description of the solution space of the recurrence on Chebyshev coefficients that is more complete than what we could find in the literature and may be of independent interest.

Note that earlier versions of the present work appeared as part of the authors' PhD theses~\cite{Benoit2012,Joldes2011,Mezzarobba2011}.

\subsection{Outline}

This article is organized as follows.
In Section~\ref{sec:series}, we review properties of Chebyshev series of D-finite functions and then study the recurrence relations satisfied by the coefficients of these series, whose use is key to the linear time complexity.
Section~\ref{sec:asympt} provides results on the asymptotics of solutions of these recurrences that will be critical for the computation of the coefficients.
The actual algorithm for this task, described in Section~\ref{sec:Clenshaw}, reminds of Fox and Parker's variant~{\cite[Chap.~5]{FoxParker1968}} of Clenshaw's algorithm~{\cite{Clenshaw1957}}.
A short description of a prototype implementation and several examples follow.

The part dedicated to the validation step starts in Section~\ref{sec:ratpolys} with a study of Chebyshev series expansions of rational functions.
Most importantly, we state remainder bounds that are then used in Section~\ref{sec:validation}, along with an enclosure method for differential equations, to validate the output of the first stage and obtain the bound~$B$.
We conclude with examples of error bounds obtained using our implementation of the validation algorithm, and some open questions.

\section{Chebyshev Expansions of D-finite Functions}\label{sec:series}

\subsection{Chebyshev Series}
\label{sec:chebseries}

Recall that the Chebyshev polynomials of the first kind are polynomials $T_n (x)
\in \mathbbm{Q} [x]$ defined for all~$n \in \mathbbm{Z}$ by the relation $T_n
(\cos \theta) = \cos (n \theta)$.
They satisfy $T_{-n}=T_n$ for all~$n$.
The family $(T_n)_{n \in \mathbbm{N}}$ is a
sequence of orthogonal polynomials over~$[- 1, 1]$ with respect to the weight
function $w (x) = 1 / \sqrt{1 - x^2}$, and hence a Hilbert basis of the
space~$L^2 (w)$. (We refer the reader to books such as
Rivlin's~\cite{Rivlin1974} or Mason and Handscomb's~\cite{MasonHandscomb2003}
for proofs of the results collected in this section.)

Expansions of functions $f \in L^2 (w)$ on this basis are known as Chebyshev
series. Instead of the more common
\begin{equation}
{\sum_n}^\prime u_n T_n = \frac{u_0}{2} T_0 + u_1 T_1 + u_2 T_2 + \cdots,
\end{equation}
we write Chebyshev series as
\begin{equation}
  \sum_{n = - \infty}^{\infty} c_n T_n (x),
  \hspace{2em} c_{-n} = c_n.
  \label{eq:chebseries}
\end{equation}
This choice makes the link between Chebyshev and Laurent expansions as well as
the action of recurrence operators on the~$c_n$ (both discussed
below) more transparent.
The \emph{Chebyshev coefficients} $c_n = \frac{1}{2} u_n$ of the expansion
of a function~$f$ are given by
\begin{equation} \label{eq:chebints}
  c_n = \frac{1}{\pi}  \int_{-1}^1 \frac{f (x) T_n (x
)}{\sqrt{1 - x^2}} \mathrm{d} x.
\end{equation}
for all $n\in\mathbbm Z$. The series~\eqref{eq:chebseries}
converges to~$f$ in the $L^2(w)$ sense for all $f \in L^2(w)$ (see for example~\cite[Chap.~5.3.1]{MasonHandscomb2003}). We denote by $\pi_d : f \mapsto
\sum_{n = - d}^d c_n T_n$ the associated orthogonal projection on the subspace
of polynomials of degree at most~$d$. \label{def:proj}

Now assume that~$f$ is a solution of Equation~(\ref{eq:deq}). As such, it may
be analytically continued to any domain~$U \subset \mathbbm{C}$ that does
not contain any singular point of the equation. Let
\begin{equation}
  E_r = \{x \in \mathbbm{C}: |x + \sqrt{x^2 - 1} | < r\} \label{eq:ellipses}
\end{equation}
be the largest elliptic domain with foci in~$\pm 1$ with this property. Since
the singular points are in finite number, we have $1 < r \leqslant \infty$.
The coefficients~$c_n$ then satisfy $c_n = O (\alpha^n)$ for
all~$\alpha > r^{-1}$; and the Chebyshev expansion~(\ref{eq:chebseries}) of~$f$
converges uniformly to~$f$ on~$E_r$~\cite[Theorem~5.16]{MasonHandscomb2003}.
Letting $x = \cos \theta$ and $z = e^{i \theta}$, it is not hard to see that
the~$c_n$ are also the coefficients of the (doubly infinite) Laurent expansion
of the function $\hat{f} (z) = f (\frac{z + z^{-1}}{2})$ around the unit
circle.
The transformation $x = \frac{z + z^{-1}}{2}$ sending~$f(x)$ to~$\hat f(z)$
is known as the inverse Joukowski transform.
\label{def:Joukowski}
It maps the elliptic disk~$E_r$ to the annulus
\[ A_r = \{ z \in \mathbbm{C}: r^{-1} < \mathopen| z \mathclose| < r \}. \]
The formula $T_n(\cos\theta)=\cos(n\theta)$ translates into
$T_n(\frac{z+z^{-1}}{2})=\frac{z^n+z^{-n}}{2}$. The coefficients~$c_n$ are
also related to those of the Fourier cosine expansion of $\theta \mapsto
f(\cos\theta)$.

\label{def:Cvgseq}
Let~$\Cvgseq \subset \mathbbm{C}^{\mathbbm{Z}}$ be the vector space of doubly
infinite sequences~$(c_n)_{n \in \mathbbm{Z}}$ such that
\[(\forall n \in \mathbbm{N}) (c_n = c_{-n})
   \hspace{2em} \text{and} \hspace{2em} (\exists \alpha < 1)
   (c_n = O_{n\to\infty}(\alpha^n)). \]
The sequence of Chebyshev coefficients of a function~$f$ that is analytic
on some complex neighborhood of~$[- 1 , 1]$ belongs to~$\Cvgseq$.
Conversely, for all~$c \in \Cvgseq$, the function series $\sum_{n = -
\infty}^{\infty} c_n T_n (x)$ converges uniformly on (some
neighborhood of) $[- 1 , 1]$ to an analytic function~$f (x
)$.

Truncated Chebyshev series are \emph{near-minimax} approximations: indeed,
they satisfy~\cite[Theorem~16.1]{Trefethen2013}
\begin{equation}
  \| f - \pi_d (f) \|_{\infty} \leqslant \Bigl(\frac{4}{\pi^2}
  \ln (d+1) + 4\Bigr)  \| f - p^{\ast}_d \|_{\infty}
  \label{eq:Lebesgue}
\end{equation}
where~$p_d^{\ast}$ is the polynomial of degree at most~$d$ that
minimizes~$\| f - p \|_{\infty}$.
\label{def:minimax}

Even though~$p^{\ast}_d$ itself can be computed to arbitrary precision using
the Remez algorithm~{\cite[Chap.~3]{Cheney1998}}, Equation~\eqref{eq:Lebesgue}
shows that we do not lose much by replacing it by~$\pi_d (f)$. Moreover, tighter
approximations are typically hard to \emph{validate} without resorting to
intermediate approximations of higher
degree~{\cite{ChevillardHarrisonJoldesLauter2011}}. The need for such
intermediate approximations is actually part of the motivation that led to the
present work. There exist a variety of other near-minimax approximations with
nice analytical properties, e.g., Chebyshev interpolation polynomials.
Our choice of truncated Chebyshev expansions is based primarily on the existence
of a recurrence relation on the coefficients~$(c_n)$ when~$f$ is a D-finite
function.

\subsection{The Chebyshev Recurrence Relation}
\label{sec:chebrec}

The polynomials $T_n$ satisfy the three-term recurrence
\begin{equation}
  2 xT_n (x) = T_{n - 1} (x) + T_{n + 1} (x
 ), \label{eq:recTn}
\end{equation}
as well as the mixed differential-difference relation
\begin{equation}
  2 (1 - x^2) T'_n (x) = n (T_{n - 1} (x) - T_{n + 1} (x))
  \label{eq:mixedrelTn}
\end{equation}
which translates into the integration formula $2 n c_n = c'_{n - 1} - c'_{n +
1}$ where $\sum c_n' T_n = (\sum c_n T_n)'$. From these equalities follows
the key ingredient of the approach developed in this article, namely that the
Chebyshev coefficients of a D\mbox{-}finite function obey a
linear recurrence with polynomial coefficients. This fact was observed by
Fox and Parker~{\cite{Fox1962,FoxParker1968}} in special cases and later
proved in general by Paszkowski~{\cite{Paszkowski1975}}. Properties of this
recurrence and generalizations to other orthogonal polynomial bases were
explored in a series of papers by Lewanowicz starting~1976 (see in
particular~{\cite{Lewanowicz1976,Lewanowicz1991}}). The automatic determination
of this recurrence in a symbolic computation system was first studied by
Geddes~\cite{Geddes1977a}.

The following theorem summarizes results regarding this recurrence, extracted
from existing
work~{\cite{Paszkowski1975,Lewanowicz1976,Lewanowicz1991,Rebillard1998,BenoitSalvy2009}}
and extended to fit our purposes. Here and in the sequel, we denote
by~$\mathbbm{Q}(n) \langle S, S^{-1} \rangle$\label{def:shift} the
skew Laurent polynomial ring over~$\mathbbm{Q}(n)$ in the
indeterminate~$S$, subject to the commutation rules
\begin{equation}
  S \lambda = \lambda S \hspace{1em} (\lambda \in \mathbbm{Q}),
  \hspace{2em} Sn = (n + 1) S. \label{eq:commutation}
\end{equation}
Likewise, $\mathbbm{Q} [n] \langle S, S^{-1} \rangle
\subset \mathbbm{Q} (n) \langle S, S^{-1} \rangle$ is
the subring of noncommutative Laurent polynomials in~$S$ themselves with
polynomial coefficients. The elements of~$\mathbbm{Q} [n]
\langle S, S^{-1} \rangle$ identify naturally with linear
recurrence operators through the left action of~$\mathbbm{Q} [n]
\langle S, S^{-1} \rangle$ on~$\mathbbm{C}^{\mathbbm{Z}}$ defined
by $(n \cdot u)_n = nu_n$ and $(S \cdot u)_n = u_{n
+ 1}$. Recall that~$L$ denotes the differential operator appearing in
Equation~(\ref{eq:deq}).

\begin{theorem}
  {\cite{Paszkowski1975,Lewanowicz1976,Lewanowicz1991,Rebillard1998,BenoitSalvy2009}}\label{thm:rec}
  Let $u, v$ be analytic functions on some complex neighborhood of the segment
  $[- 1, 1]$, with Chebyshev expansions
  \[u (x) = \sum_{n = - \infty}^{\infty} u_n T_n (x
    ) {,} \hspace{1em} v (x) = \sum_{n = -
     \infty}^{\infty} v_n T_n (x) . \]
  There exist difference operators~$P, Q \in \mathbbm{Q} [n]
  \langle S, S^{-1} \rangle$ with the following properties.
  \begin{enumerate}
    \item \label{item:thmdeqrec}The differential equation $L \cdot u (x
   ) = v (x)$ is satisfied if and only if
    \begin{equation}
      P \cdot (u_n) = Q \cdot (v_n) .
      \label{eq:chebrec}
    \end{equation}
    \item \label{item:P}The left-hand side operator~$P$ is of the form $P =
    \sum_{k = - s}^s b_k (n) S^k$ where $s = r + \max_i  (
    \deg a_i)$ and $b_{-k} (- n) = - b_k (n
   )$ for all~$k $.

    \item \label{item:Q}Letting
    \begin{equation}
      \delta_r (n) = 2^r  \prod_{i = - r + 1}^{r - 1} (n -
      i), \hspace{2em} I = \frac{1}{2 n}  (S^{-1} - S),
      \label{eq:delta I}
    \end{equation}
    we have $Q = Q_r = \delta_r (n) I^r$ (this expression is to be interpreted as a polynomial identity in $\mathbbm{Q} (n) \langle S, S^{-1} \rangle$).
    In particular, $Q$~depends only on~$r$ and satisfies the same symmetry property as~$P$.
  \end{enumerate}
\end{theorem}

Note that~$I$, as defined in Eq.~\eqref{eq:delta I}, may be interpreted as an
operator from the \emph{symmetric} sequences
$(u_{\mathopen|n\mathclose|})_{n\in\mathbbm Z}$ to the sequences $(u_n)_{n \in
\mathbbm Z \setminus \{0\}}$ defined only for nonzero~$n$. A sloppy but perhaps
more intuitive statement of the main point of Theorem~\ref{thm:rec} would be:
``$(\int)^r L \cdot u = w$ if and only if $\delta_r(n) P \cdot u = w$, up to
some integration constants''.

\begin{proof}
Assume $L \cdot u = v$. Benoit and Salvy~\cite[Theorem~1]{BenoitSalvy2009} give
a simple proof that~\eqref{eq:chebrec} holds for some~$P, Q \in \mathbbm{Q} (n)
\langle S, S^{-1} \rangle$. The fact that~$P$ and~$Q$ can actually be taken to
have polynomial coefficients and satisfy the properties listed in the last two
items then follows from the explicit construction discussed in Section~4.1 of their
article, based on Paszkowski's algorithm~\cite{Paszkowski1975,Lewanowicz1976}.
More precisely, multiplying both members of~\cite[Eq.~(17)]{BenoitSalvy2009}
by~$\delta_r(n)$ yields a recurrence of the prescribed form. The recurrence has
polynomial coefficients since $\delta_r(n) I^r \in \mathbb{Q}\langle S, S^{-1}
\rangle$. Rebillard's thesis~{\cite[ Section~4.1]{Rebillard1998}} contains detailed
proofs of this last observation and of all assertions of Item~\ref{item:P}.
Note that, although Rebillard's and Benoit and Salvy's works are
closest to the formalism we use, several of these results actually go back
to~{\cite{Paszkowski1975,Lewanowicz1976,Lewanowicz1991}}.

There remains to prove the ``if'' direction. Consider sequences $u, v \in
\Cvgseq$ such that  $P \cdot u = Q \cdot v$, and let~$y \in \Cvgseq$ be the Chebyshev coefficient sequence of the (analytic) function~$L \cdot u$.
We then have $P
\cdot u = Q \cdot y$ by the previous argument. This implies $Q \cdot v = Q
\cdot y$, whence finally $y=v$ by Lemma~\ref{lemma:injective} below.
\end{proof}

\begin{lemma} \label{lemma:injective}
The restriction to~$\Cvgseq$ of the operator~$Q$ from Theorem~\ref{thm:rec} is
injective.
\end{lemma}

\newcommand\myflat[1]{#1^{\flat}}
\begin{proof}
With the notation of Theorem~\ref{thm:rec}, we show by induction on~$r
\geqslant 1$ that
\begin{equation}
  (v \in \Cvgseq) \wedge \bigl(\mathopen| n \mathclose| \geqslant r
  \implies (Q_r \cdot v)_n = 0\bigr) \implies v = 0.
  \label{eq:ind hyp chebrec}
\end{equation}
First, we have $(\ker Q_1) \cap \Cvgseq = \{ 0 \}$ since any sequence belonging
to~$\Cvgseq$ converges to zero as $n \to \pm \infty$.
Now assume that~\eqref{eq:ind hyp chebrec} holds, and let~$v \in \Cvgseq$ be
such that $(Q_{r + 1} \cdot v)_n = 0$ for $\mathopen|n\mathclose| \geqslant
r+1$. Let $w = Q_r \cdot v$. Observe that~$\Cvgseq$ is stable under the action of~$\mathbb Q(n)\langle S,S^{-1}\rangle$, so~$w \in \Cvgseq$. Since~$r \geqslant
1$, we have
\begin{align*}
     2n\, Q_{r + 1} & = \delta_{r + 1} (n)  (S^{-1} - S) I^r\\
     & = 2 \, ((n + r)  (n + r - 1) S^{-1} \delta_r (n) - (n - r)  (n - r + 1) S \delta_r (n)) I^r\\
     & = 2 \, ((n + r)  (n + r - 1) S^{-1} - (n - r)  (n - r + 1) S) Q_r.
\end{align*}
Hence, for $\mathopen| n \mathclose| \geqslant r + 1$, it holds that
\begin{equation}
  (n + r)  (n + r - 1) w_{n - 1} = (n - r
 )  (n - r + 1) w_{n + 1} . \label{eq:rec w}
\end{equation}
Unless~$w_n$ is ultimately zero, this implies that~$w_{n + 1} / w_{n - 1} \to
1$ as $n \to \infty$, which contradicts the fact that~$w \in \Cvgseq$. It
follows that~$w_n = 0$ for~$\mathopen| n \mathclose|$ large enough, and,
using~(\ref{eq:rec w}) again, that~$w_n = 0$ as soon as~$\mathopen| n
\mathclose| \geqslant r$. Applying the hypothesis~(\ref{eq:ind hyp chebrec})
concludes the induction.
\end{proof}

An easy-to-explain way of computing a recurrence of the form~\eqref{eq:chebrec}
is as follows. We first perform the change of variable $x = \frac{1}{2} (z +
z^{-1})$ in the differential equation~\eqref{eq:deq}. Then, we compute a
recurrence on the Laurent coefficients of~$\hat{u} (z) = u (x)$ by the
classical (Frobenius) method.

\begin{example} \label{ex:algsubs}
The function $y(x)=\arctan(x/2)$ satisfies the homogeneous equation $(x^2+4) \, y''(x) + 2 x \, y'(x) = 0$. The substitutions
\[
    x = \frac{z+z^{-1}}{2}
    \qquad
    \frac{\d}{\d x} = \frac{2z}{z-z^{-1}} \frac{\d}{\d z}
\]
yield (after clearing common factors and denominators)
\[
    (z+1)(z-1)(z^4+18z^2+1) \, \hat y''(z)
    + 2(z^4-2z^2-19)z \, \hat y'(z)
    =0.
\]
We then set $\hat y(z)=\sum_{n=-\infty}^{\infty} c_n z^n$ and extract the
coefficient of~$z^n$ (which amounts to replacing $z$~by~$S^{-1}$ and
$z\frac{\d}{\d z}$~by~$n$) to get the recurrence
\begin{multline*}
(n-2)(n-3)c_{n-3}+(n-1)(17n-38)c_{n-1} \\
-(n+1)(17n+38)c_{n+1}-(n+2)(n+3)c_{n+3} = 0.
\end{multline*}
\end{example}

Benoit and Salvy~{\cite{BenoitSalvy2009}} give a unified presentation of
several alternative algorithms, including Paszkowski's, by interpreting them as
various ways to perform the substitution $x \mapsto \frac{1}{2}  (S + S^{-
1})$, $\frac{\mathrm{d}}{\mathrm{d} x} \mapsto (S - S^{-1})^{-1}  (2 n)$ in a
suitable non-commutative division algebra. In our setting where the
operator~$L$ is nonsingular over~$[-1,1]$, they prove that all these algorithms
compute the same operator~$P$.

\begin{remark}
As applied in Example~\ref{ex:algsubs}, the method based on setting $x=\frac12(z+z^{-1})$ in the differential equation does not always yield the same operator as Paszkowski's algorithm.
It can be modified to do so as follows: instead of clearing the denominator of the differential equation in~$z$ given by the rational substitution, move this denominator to the right-hand side, translate both members into recurrences, and then \emph{remove a possible common left divisor} of the resulting operators $P, Q \in \Q(n) \langle S, S^{-1}\rangle$.
\end{remark}

\begin{definition}
\label{def:chebrec}
Following Rebillard, we call the recurrence relation~\eqref{eq:chebrec}
computed by Paszkowski's algorithm (or any equivalent method) the
\emph{Chebyshev recurrence} associated to the differential
equation~\eqref{eq:deq}.
\end{definition}

\begin{remark}
\label{rk:symmetry}
By Theorem~\ref{thm:rec}(\ref{item:P}) and with its notation, for any
sequence~$(u_n)_{n\in\Z}$, we have the equalities
\[
  \forall n, \quad
    \sum_k b_k(n) u_{n+k}
    = - \sum_k b_{-k}(-n) u_{n+k}
    = - \sum_k b_k(-n) u_{-n-k},
\]
that is, $P \cdot (u_n)_{n \in \mathbbm Z} = - P \cdot (u_{-n})_{n \in \mathbbm Z}$. In particular, if $(u_n)_{n \in \Z}$ is a solution of a homogeneous
Chebyshev recurrence, then so is $(u_{-n})_{n \in \Z}$, and $(u_n+u_{-n})$ is a \emph{symmetric} solution.
Not all solutions are symmetric.
For instance, the differential equation $y'(x)=xy(x)$ corresponds to the recurrence $-c_{n-2}+4n\,c_n+c_{n+2}=0$ which allows for $u_{-2}=3, u_{-1}=12, u_0=1, u_1=2, u_2=3$.
\end{remark}

\subsection{Solutions of the Chebyshev Recurrence}

Several difficulties arise when trying to use the Chebyshev recurrence to
compute the Chebyshev coefficients.

A first issue is related to initial
conditions. Here it may be worth contrasting the situation with the more
familiar case of the solution of differential equations in power series. Unlike
the first few Taylor coefficients of~$y$, the Chebyshev coefficients $c_0, c_1,
\ldots$ that could serve as initial conditions for the recurrence are not
related in any direct way to initial or boundary conditions of the differential
equation. In particular, as can be seen from Theorem~\ref{thm:rec} above, the
order $2 s$ of the recurrence is larger than that of the differential equation
except for degenerate cases. Hence we need to somehow ``obtain more
initial values for the recurrence than we naturally have at hand''
\footnote{Nevertheless, the recurrence~\eqref{eq:chebrec} shows that the Chebyshev
coefficients of a D-finite function are rational linear combinations of a
finite number of integrals of the form~\eqref{eq:chebints}. Computing these
coefficients efficiently with high accuracy is an interesting problem to which we
hope to come back in future work. See Benoit~\cite{Benoit2012} for some results.}.

Next, also in contrast to the case of power series, the leading and trailing
coefficients~$b_{\pm s}$ of the recurrence~\eqref{eq:chebrec} may vanish for
arbitrarily large values of~$n$ even though the differential
equation~(\ref{eq:deq}) is nonsingular. The zeroes of~$b_s(n-s)$ are called the
\emph{leading singularities} of~\eqref{eq:chebrec}, those of~$b_{-s}(n+s)$, its
\emph{trailing singularities}. In the case of Chebyshev recurrences, leading and
trailing singularity sets are opposite of each other.

One reason for the presence of (trailing) singularities is clear: if a
polynomial $y = \sum y_{\abs{n}} T_n$ of degree~$d$ is a solution of $L \cdot y
= 0$, then necessarily $b_{-s}(d+s)=0$. However, even differential equations
without polynomial solutions can have arbitrarily large leading and trailing
singularities, as shown by the following example.

\begin{example}
  For all $k \in \mathbbm Z$, the Chebyshev recurrence relation associated to
  the differential equation
  $y''(x) + (x^2+1)\,y'(x) - k\,x\,y(x)=0$,
  namely
  \[
    \begin{split}
      (n+1)(n-k-3) \, c_{n-3}
      + (n-1)(5n+k+7) \, c_{n-1}
      + 8 n (n+1) (n-1) \, c_n \\
      - (n+1)(5n-k-7) \, c_{n+1}
      - (n-1)(n+k+3) \, c_{n+3}
      = 0,
    \end{split}
  \]
  admits the leading singularity~$n=k$. For $k=1$, the differential equation
  has no polynomial solution.
\end{example}

We do however have some control over the singularities.

\begin{proposition} \label{prop:sing}
  With the notations of Theorem~\ref{thm:rec}, the coefficients of the
  Chebyshev recurrence satisfy the relations
  \begin{equation}
    b_{j - i} (- j) = - b_{j + i} (- j),
    \hspace{2em} \mathopen| j \mathclose| \leqslant r - 1, \hspace{1em} i \in
    \mathbbm{N}, \label{eq:prop sing}
  \end{equation}
  with~$b_k = 0$ for~$\mathopen| k \mathclose| > s$.
  In particular, $b_s(n)$ is zero for all $n \in \llbracket 1, r-1 \rrbracket$.
\end{proposition}

\begin{proof}
We proceed by induction on~$r$. When~$j = 0$,
assertion~(\ref{eq:prop sing}) reduces to $b_{-i} (0) = - b_i
(0)$, which follows from the second item of
Theorem~\ref{thm:rec}. In particular, this proves the result for~$r = 1$. Now
let~$r \geqslant 2$ and assume that the proposition holds when~$L$ has
order~$r - 1$. Write $L = \myflat{L} + \partial^r p_r (x)$
where~$p_r \in \mathbbm{Q} [x]$ and~$\myflat{L}$ is a differential
operator of order at most~$r - 1$. Letting~$\myflat{P} = \sum_{k \in \mathbbm{Z}}
\myflat{b}_k (n) S^k$ be the Chebyshev recurrence operator
associated to~$\myflat{L}$, we then have~{\cite{BenoitSalvy2009}}
\begin{equation}
  \delta_r (n)^{-1} P = I \delta_{r - 1} (n)^{-1} 
  \myflat{P} + p_r (\tfrac{1}{2}  (S + S^{-1}))
  \label{eq:P hatP}
\end{equation}
where the last term denotes the evaluation of~$p_r$ at $x = \frac{1}{2} 
(S + S^{-1})$. Since
\[I \delta_{r - 1} (n)^{-1} = (n \delta_r (n
  ))^{-1}  ((n - r + 2)  (n - r + 1
  ) S^{-1} - (n + r - 2)  (n + r - 1) S
  ) \]
by the commutation rule~(\ref{eq:commutation}), relation~(\ref{eq:P hatP}) rewrites as
\begin{align*}
  P &=
       \frac{1}{n}\sum_k \bigl((n - r + 2)(n - r + 1)\myflat{b}_{k + 1}(n - 1) \\
      &\hspace{4em}   - (n + r - 2)  (n + r - 1)  \myflat{b}_{k
       - 1} (n + 1)\bigr) S^k \\
     & + \delta_r (n) p_r (\frac{1}{2}(S + S^{-1})) .
\end{align*}
The case~$j = 0$ having already been dealt with, assume~$0 < \mathopen| j \mathclose|
< r$. Since~$\delta_r (- j) = 0$ and~$p_r$ is a polynomial, it
follows by extracting the coefficient of~$S^k$ in the last equality and
evaluating at~$n = - j$ that
\begin{multline}
     - jb_k (- j) = (j + r - 2)  (j + r - 1
    )  \myflat{b}_{k + 1} (- j - 1)\\
     - (j - r + 2)  (j - r + 1)  \myflat{b}_{k - 1}
     (- j + 1) .
\end{multline}
Now~$\myflat b_{j - i} (- j) = - \myflat b_{j + i} (- j)$ for
$\mathopen| j \mathclose| < r - 1$ by the induction hypothesis, and the term
involving~$\myflat{b}_{k \pm 1}$ vanishes for $j = \mp (r - 1)$ and
$j = \mp (r - 2)$. In each case, we obtain $b_{j - i} (- j
) = -b_{j + i} (- j)$.
\end{proof}

\begin{corollary}
\label{cor:symmetry}
Let~$P$ be the Chebyshev recurrence operator associated to~$L$. The image
by~$P$ of a symmetric sequence~$(u_{\abs n})_{n\in\Z}$ satisfies $(P \cdot
u)_n=0$ for~$\abs{n} < r$.
\end{corollary}

\begin{proof}
Since
\[
  (P \cdot u)_n = \sum_{k \in \mathbbm{Z}} b_k(n) u_{n+k} = \sum_{i \in
  \mathbbm{Z}} b_{i-n}(n) u_i,
\]
it follows from Proposition~\ref{prop:sing} with~$j = - n$ and
$\mathopen|n\mathclose| < r$ that
\[
  \sum_{i \in \mathbbm{Z}} b_{i - n} (n) u_i
  = -\sum_{i \in \mathbbm{Z}} b_{- i - n} (n) u_i
  = -\sum_{i \in \mathbbm{Z}} b_{i - n} \left( n \right) u_i,
\]
that is, $\left( P \cdot u \right)_n = - \left( P \cdot u \right)_n$.
\end{proof}

Last but not least, Chebyshev recurrences always admit \emph{divergent} solution sequences.
Divergent solutions do not correspond to the expansions of solutions of the differential equation the recurrence comes from.

\begin{example}
The Chebyshev recurrence associated to the equation~$y'=y$ is
\[ (P \cdot u)_n = u(n+1) + 2n \, u(n) - u(n-1) = 0. \]
In terms of the modified Bessel functions $I_\nu$~and~$K_\nu$, a basis of solutions of the recurrence is given by the sequences
$(I_\nu(1))_{\nu\in\Z}$~and~$(K_\nu(1))_{\nu\in\Z}$. The former is the
coefficient sequence of the Chebyshev expansion of the exponential function and
decreases as~$\Theta(2^{-\nu}\,\nu!^{-1})$. The later satisfies
$K_\nu(1)=\Theta(2^\nu\,(\nu-1)!)$.
\end{example}

\section{Convergent and Divergent Solutions}
\label{sec:asympt}

\subsection{Elements of Birkhoff-Trjitzinsky Theory}

Before studying in more detail the convergent and divergent solutions of the
Chebyshev recurrence relation, we recall some elements of the asymptotic
theory of linear difference equations. Much of the presentation is based on
Wimp's book~{\cite[Appendix~B]{Wimp1984}}, to which we refer the reader for
more information.

\begin{definition}
  For all $\rho \in \mathbbm{N}\backslash \{ 0 \}, J \in \mathbbm{N}$, $\kappa
  \in \mathbbm{Q}$, $\alpha, \pi_j, \theta, \beta_{j, i} \in \mathbbm{C}$, we
  call the formal expansion
  \begin{equation}
    \bar{u} (n) = n!^{\kappa} \alpha^n e^{\pi (n)}  \sum_{j = 0}^J (\ln
    n)^j  \sum_{i = 0}^{\infty} \beta_{j, i} n^{\theta - i / \rho}
    \label{eq:FAS}
  \end{equation}
  where
  \[ \pi (n) = \pi_1 n^{1 / \rho} + \cdots + \pi_{\rho - 1} n^{(\rho - 1) /
     \rho} \]
  a {\emph{formal asymptotic series (FAS)}}. The set of all FAS is denoted
  by~$\mathcal{B}$.
\end{definition}

Formal asymptotic series are to be interpreted as asymptotic expansions of
sequences as $n \rightarrow \infty$. The product of two FAS is defined in the
obvious way and is again an FAS. The same goes for the substitution $n \mapsto
n + k$ for fixed $k \in \mathbbm{Z}$, using identities such as $(n +
k)^{\theta} = n^{\theta}  (1 + k \theta n^{- 1} + \cdots)$. The sum of two
FAS is not always an FAS, but that of two FAS sharing the same parameters
$\kappa, \alpha, \pi$ is. Thus, it makes sense to say that an FAS $\bar{u} \in
\mathcal{B}$ satisfies a recurrence
\begin{equation}
  \bar{b}_s (n)  \bar{u} (n+s) + \cdots + \bar{b}_0 (n)  \bar{u} (n) = 0
  \label{eq:formalrec}
\end{equation}
with formal series coefficients of the form
\begin{equation}
  \bar{b}_k (n) = n^{\tau_k / \omega}  (\beta_{k, 0} + \beta_{k, 1} n^{- 1 /
  \omega} + \beta_{k, 2} n^{- 2 / \omega} + \cdots) \in \mathbbm{C} ((n^{- 1 /
  \omega})) . \label{eq:coeffasympt}
\end{equation}
Also, given $s$~FAS $\bar{u}_0, \ldots, \bar{u}_{s - 1} \in \mathcal{B}$, the
Casoratian
\[ C (n) = \det (\bar{u}_j (n + i))_{0 \leqslant i, j < s} \]
belongs to~$\mathcal{B}$ as well.

Following Wimp, we say that $\bar{u}_1, \ldots, \bar{u}_s \in \mathcal{B}$ are
{\emph{formally linearly independent}} when their Casoratian is nonzero. Note
that the elements of any subset of $\{ \bar{u}_1, \ldots, \bar u_s \}$
are then formally linearly independent as well. Indeed, it can be checked by
induction on~$s$ that $s$~FAS $\bar{u}_1, \ldots, \bar{u}_s$ are formally
linearly dependent if and only if there exists a relation of the form
$\bar{\mu}_1 (n)  \bar{u}_1 (n) + \cdots + \bar{\mu}_n (s)  \bar{u}_s (n)
= 0$ where the~$\bar{\mu}_k$ are FAS such that{\footnote{Like Wimp, but unlike
most authors, we consider recurrences rather than difference equations.
Accordingly, we forbid factors of the form $e^{\pi_{\rho} n}$ with $|
\operatorname{Im} \pi_{\rho} | > \pi$ in~\eqref{eq:FAS}, so that the~$\mu_k (n)$ are
actually constants in our setting.}} $\bar{\mu}_k (n + 1) = \bar{\mu}_k (
n)$.

\begin{definition}
  \label{def:asympt}The FAS~\eqref{eq:FAS} is said to be an {\emph{asymptotic
  expansion}} of a sequence $(u_n) \in \mathbbm{C}^{\mathbbm{N}}$, and we
  write $u_n \sim \bar{u} (n)$, when for any truncation order~$I$, the
  relation
  \[ u_n = n!^{\kappa} \alpha^n e^{\pi (n)} \sum_{j = 0}^J (\ln n)^j  \left(
     \sum_{i = 0}^{I - 1} \beta_{j, i} n^{\theta - i / \rho} + O (n^{\theta -
     I / \rho}) \right) \]
  holds as~$n \rightarrow \infty$.
\end{definition}

The following fundamental result is known as the Birkhoff-Trjitzinsky theorem,
or ``main asymptotic existence theorem'' for linear recurrences. It will be
the starting point of our analysis of the computation of ``convergent''
solutions of the Chebyshev recurrence by backward recurrence.

\begin{theorem}
  {\cite{Birkhoff1930,BirkhoffTrjitzinsky1933,Turrittin1960,Immink1991}}\label{thm:BT}
  Consider a linear recurrence
  \begin{equation}
    b_s (n) u_{n + s} + \cdots + b_0 (n) u_n = 0 \label{eq:BTrec}
  \end{equation}
  whose coefficients $b_0, \ldots, b_s$ admit asymptotic expansions (in the
  sense of Definition~\ref{def:asympt}) $\bar{b}_0, \ldots, \bar{b}_s$ of
  the form~\eqref{eq:coeffasympt} for some integer~$\omega \geqslant 1$. Then,
  \begin{enumerate}
    \item \label{item:BT:formal}the (formal) recurrence~\eqref{eq:formalrec}
    possesses a system of $s$~formally linearly independent FAS solutions;
    
    \item \label{item:BT:analytic}for any $s$ formally linearly independent
    solutions $\bar{e}_1, \ldots, \bar{e}_s \in \mathcal{B}$
    of~\eqref{eq:formalrec}, there exists complex sequences $e_1 = (e_{1,
    n})_{n \geqslant N}, \ldots, e_s = (e_{s, n})_{n \geqslant N}$ defined in
    some neighborhood of infinity, with the property that $e_k \sim \bar{e}_k$
    for all~$k$, and such that $(e_1, \ldots, e_s)$ is a basis of the
    solution space of~\eqref{eq:BTrec} for $n \geq N$.
  \end{enumerate}
\end{theorem}

We note that many expositions of the Birkhoff-Trjitzinsky theorem warn about
possible major gaps in its original proof. However, the consensus among
specialists now appears to be that these issues have been resolved in modern
proofs~{\cite{Immink1991,vdPutSinger1997}}. Besides, under mild additional
assumptions on the Chebyshev recurrence, all the information needed in our
analysis is already provided by the more elementary Perron-Kreuser
theorem{\footnote{The Perron-Kreuser theorem yields the existence of a basis
of solutions such that $e_{i, n + 1} / e_{i, n} \sim \alpha n^{\kappa_i}$,
under the assumption that $\kappa_i = \kappa_j \Rightarrow | \alpha_i | \neq |
\alpha_j |$. It does not require that the coefficients of~\eqref{eq:BTrec}
admit full asymptotic expansions, which makes it stronger than
Theorem~\ref{thm:BT} in some respects.}}
(cf.~{\cite{Guelfond1963,Meschkowski1959,Milne-Thomson1933}}) or its extensions
by Schäfke~{\cite{Schaefke1965}}. See also
Immink~\cite{Immink1996} and the references therein for an alternative approach in the case of recurrences with polynomial coefficients, originating in unpublished work by Ramis.

Also observe that for any subfamily $(f_1, \ldots, f_{s'})$ of the
sequences~$e_i$ from Theorem~\ref{thm:BT}, the matrix $(f_{j, n + i})_{1
\leqslant i, j \leqslant s'}$ is nonsingular for large~$n$. In particular, the
$e_{i, n}$ can vanish only for finitely many~$n$. The more precise statement
below will be useful in the sequel.

\begin{lemma}
  \label{lemma:Casorati}Assume that the sequences $(e_{0, n})_n, \ldots, (e_{s
  - 1, n})_n$ admit formally linearly independent asymptotic expansions of the
  form~\eqref{eq:FAS}, with $\alpha_i \in \mathbbm{C} \setminus \{0\}$,
  $\kappa_i \in \mathbbm{Q}$. Then the Casorati determinant
  \[ C (n) = \left|\begin{array}{cccc}
       e_{0, n} & e_{1, n} & \cdots & e_{s - 1, n}\\
       e_{0, n + 1} &  &  & e_{s - 1, n + 1}\\
       \vdots &  &  & \vdots\\
       e_{0, n + s - 1} & e_{1, n + s - 1} & \cdots & e_{s - 1, n + s - 1}
     \end{array}\right| \]
  satisfies
  \[ C (n) = \beta e_{0, n} e_{1, n} \cdots e_{s - 1, n} n^{\theta}  ((\ln
     n)^{\lambda} + O ((\ln n)^{\lambda - 1})), \hspace{2em} n \rightarrow
     \infty, \]
  for some $\beta \in \mathbbm{C}\backslash \{ 0 \}$, $\theta \in
  \mathbbm{C}$, and $\lambda \in \mathbbm{N}$.
\end{lemma}

\begin{proof}
  Write $C (n) = e_{0, n} e_{1, n} \cdots e_{s - 1, n} C' (n)$. The formal
  linear independence hypothesis means that $C (n)$, and hence $C' (n)$,
  admit nonzero FAS as asymptotic expansions. Additionally,
  \[ C' (n) = \det \left(\frac{e_{j, n + i}}{e_{j, n}} \right)_{0
     \leqslant i, j < s} \]
  has at most polynomial growth, so that the leading term of its asymptotic
  expansion must be of the form $n^{\theta}  (\ln n)^{\lambda}$.
\end{proof}

\subsection{Newton Polygon of a Chebyshev Recurrence}
\label{sec:Newton polygon}

\begin{figure}
  \begin{center}
    \begin{tikzpicture}[x=0.7cm,y=0.7cm]
      \node[coordinate] (topright) at (3.8,3.9) {};
      \node[coordinate] (axisright) at ($(0,0)!(topright)!(1,0)$) {};
      \draw[thin, dotted] ($-1*(axisright)$) grid[step=1] (topright);
      \draw[->] ($-1*(axisright)$) -- (axisright) node[right] {$S$};
      \draw[->] (0,0) -- ($(0,0)!(topright)!(0,1)$) node[above] {$n$};
      \foreach \coeff in {(-3,4-1), (-2,4-3), (0,4-1), (2,4-3), (3,4-1)}
        \fill \coeff circle (2pt);
        {}
      \draw[thick,auto]
        (-3,4-1) --
          node {$\alpha_3$}
          node[swap] {$\kappa_{3}$}
        (-2,4-3) --
          node[pos=.125] {$\alpha_2$}
          node[pos=.375] {$\alpha_1$}
          node[pos=.625] {$\alpha_{-1}$}
          node[pos=.875] {$\alpha_{-2}$}
          node[swap] {$\kappa_{2}=\cdots=\kappa_{-2}$}
        (2,4-3) --
          node {$\alpha_{-3}$}
          node[swap] {$\kappa_{-3}$}
        (3,4-1)
      ;
    \end{tikzpicture}
    \caption{The Newton polygon of a Chebyshev recurrence.}
    \label{fig:Newton polygon}
  \end{center}
\end{figure}
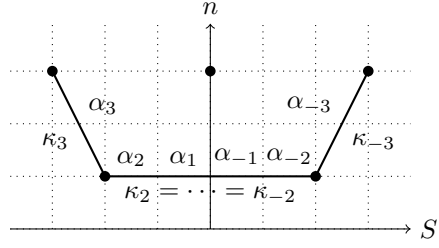

The formal solutions described in Theorem~\ref{thm:BT} may be constructed
algorithmically using methods going back to Poincaré~{\cite{Poincare1885}}
and developed by many authors. See in particular Adams~{\cite{Adams1928}} and
Birkhoff~{\cite{Birkhoff1930}} for early history, Tournier~{\cite{Tournier1987}} for a
comparison of several methods from a Computer Algebra perspective, and Balser
and Bothner~{\cite{BalserBothner2010}} for a modern algorithm as well as more
references.

Here, we are mostly interested in the parameters $\kappa$~and~$\alpha$ that
control the ``exponential'' growth rate of the solutions. We briefly recall
how the possible values of these parameters are read off the recurrence using
the method of Newton polygons. Consider again the Chebyshev recurrence
operator
\[ P = b_{- s} (n) S^{- s} + \cdots + b_0 (n) + \cdots + b_s (n) S^s \]
from Section~\ref{sec:chebrec}. The {\emph{Newton polygon}} of~$P$ is defined
as the lower convex hull of the points $p_k = (k, - \deg b_k) \in
\mathbbm{R}^2$ (see Figure~\ref{fig:Newton polygon}). To each edge $[ p_i,
p_j]$ $(i < j)$ of the polygon is attached a \emph{characteristic equation}
\[ \chi_i (\alpha) = \sum_{k : p_k \in [ p_i, p_j]} \operatorname{lc} (b_k)
   \alpha^{k - i}, \]
where $\operatorname{lc} (p)$ denotes the leading coefficient of~$p$. Note that the
degrees of the $\chi_i$ sum to~$2 s$. Let
\[ \alpha_s, \alpha_{s - 1}, \ldots, \alpha_1, \alpha_{- 1}, \ldots, \alpha_{-s +1}, \alpha_{- s} \]
be the sequence of all roots of the polynomials~$\chi_i$, with multiplicities,
the roots corresponding to distinct edges being written in the order of
increasing~$i$ and the roots of each~$\chi_i$ in that of increasing modulus.
For all~$k$, let $\kappa_k$ be the slope of the edge associated to~$\alpha_k$.
(Thus, each $\kappa_k$ is repeated a number of times equal to the horizontal
length of the corresponding edge, and we have
$\kappa_s \leqslant \kappa_{s - 1} \leqslant \cdots \leqslant \kappa_{- s}$.)

How does this relate to the asymptotics of Chebyshev series? Assume
that~$n!^{\kappa} \alpha^n$ is the leading factor of some FAS
solution~$\bar{u}$ of $P \cdot \bar{u} = 0$. It is not too hard to see that,
in order for asymptotically dominant terms of $P \cdot \bar{u}$ to cancel out,
$\kappa$~must be among the slopes of the Newton polygon of~$P$, and
$\alpha$~must be a root of the characteristic equation of the corresponding
edge. This gives all possible values of $\kappa$~and~$\alpha$. Conversely, the
construction behind Theorem~\ref{thm:BT}~\eqref{item:BT:formal} yields a
number of linearly independent FAS with given $\kappa$~and~$\alpha$ equal to the
multiplicity of~$\alpha$ as a root of the characteristic equation of the edge
of slope~$\kappa$. In the case of Chebyshev recurrences, the Newton polygon
has the following symmetry property.

\begin{proposition}
  \label{prop:Newton polygon}The slopes~$\kappa_i$ of the Newton polygon
  of~$P$ and the roots~$\alpha_i$ of its characteristic equations satisfy
  $\kappa_{- i} = - \kappa_i$ and $| \alpha_{- i} | = | \alpha_i |^{- 1}$ for
  all~$i$. In addition, none of the roots associated to the horizontal edge
  (if there is one) has modulus~$1$.
\end{proposition}

\begin{proof}
  By Theorem~\ref{thm:rec}, the coefficients~$b_k$ of~$P$ are related by $b_{-
  k} (n) = - b_k (- n)$. Hence, the Newton polygon is symmetric with respect
  to the vertical axis, and $\kappa_{- i} = - \kappa_i$ for all~$i$. Now
  fix~$i$, and let $\varepsilon_i = [ p_{\ell (i)}, p_{r (i)}]$ be the edge
  of slope~$\kappa_i$. The characteristic equation of~$\varepsilon_i$ reads
  \[ \chi_i (\alpha) = \sum_{k : p_k \in \varepsilon_i} \operatorname{lc} (b_k)
     \alpha^{k - \ell (i)} = \sum_{k : p_k \in \varepsilon_i} (- 1)^{1 +
     \deg b_k} \operatorname{lc} (b_{- k}) \alpha^{k - \ell (i)}, \]
  where $\operatorname{lc} (b)$ denotes the leading coefficient of a polynomial~$b$.
  Using the relation $\deg b_k - \deg b_{\ell (i)} = \kappa_i  (k - \ell (
  i))$ for $p_k$ lying on $\varepsilon_i$, we get
  \begin{align*}
    {\chi}_{i}({\alpha})     &=
      {\pm}\sum_{k:p_{k}{\in}{\varepsilon}_{-i}}(-1)^{{\kappa}_{i} (-k+{\ell}(-i))} \operatorname{lc}(b_{k}) {\alpha}^{-k+{\ell}(-i)}\\
         &=
    {\pm}{\alpha}^{{\ell}(i)-{\ell}(-i)} {\chi}_{-i}((-1)^{-{\kappa}_{i}} {\alpha}^{-1}),
  \end{align*}
  and hence~$| \alpha_{- i} | = | \alpha_i |^{- 1}$.
  
  There remains to prove that $\kappa_i = 0$ implies $| \alpha_i | \neq 1$.
  Under the change of variable $x = \frac{1}{2}  (z + z^{- 1})$, the leading
  term with respect to~$\theta = z \frac{\mathd}{\mathd z}$ of $(
  \frac{\mathd}{\mathd x})^k$ is $2^k  (z - z^{- 1})^{- k}$. (The leading term
  is well-defined because the commutation relation between $z$~and~$\theta$
  preserves degrees.) Therefore, the characteristic equation associated to
  the slope~$\kappa = 0$ (when there is one) of the recurrence operator~$P_1$ 
  obtained by changing $z$ into~$S^{- 1}$ and $\theta$ into~$n$ is
  \[ \chi_{\operatorname{horiz}} (\alpha) := \sum_{k : \deg p_k = \max_i \deg
     p_i} \operatorname{lc} (b_k) \alpha^{k - i} = a_r \left(\frac{\alpha +
     \alpha^{- 1}}{2} \right), \]
  where $a_r$~is the leading coefficient of~\eqref{eq:deq}. Since $P$~is a
  right factor of~$P_1$, the characteristic polynomial associated to~$\kappa =
  0$ in the Newton polygon of~$P$ divides~$\chi_{\operatorname{horiz}}$. But,
  due to the assumptions stated in Section~\ref{sec:setting},
  the polynomial~$a_r(x)$ does not vanish for $x \in [ - 1, 1]$,
  hence $\chi_{\operatorname{horiz}} (\alpha) \neq 0$ for $| \alpha | = 1$.
\end{proof}

Summing up, the asymptotic structure of the solutions of the Chebyshev
recurrence may be described as follows. Similar observations were already made
by Rebillard~{\cite[Chap.~5]{Rebillard1998}}.

\begin{corollary}
  \label{cor:asympt sols}For large enough~$N$, the space of sequences $(
  u_n)_{n \geqslant N}$ satisfying $(P \cdot u)_n = 0$ (``germs of solution
  at infinity of the Chebyshev recurrence'') has a basis comprising
  $s$~convergent sequences $e_1, \ldots, e_s$ and $s$~divergent
  sequences~$e_{- 1}, \ldots, e_{- s}$, all with formally linearly independent
  FAS expansions, such that
  \[ e_{i, n} = n!^{\kappa_i} \alpha_i^n e^{o (n)}, \hspace{2em} n
     \rightarrow \infty . \]
  In particular, we have $\ln | e_{i, n} e_{- i, n} | = o (n)$ for all~$i$.
\end{corollary}

\begin{proof}
  This follows from Theorem~\ref{thm:BT}, Proposition~\ref{prop:Newton
  polygon}, and the description of a basis of formal solutions at infinity
  using the Newton polygon.
\end{proof}

\section{Computing Approximation Polynomials}
\label{sec:Clenshaw}

\subsection{Clenshaw's Algorithm Revisited}

At this point, we know that Chebyshev expansions of D\mbox{-}finite functions
correspond to the symmetric, convergent solutions of the Chebyshev recurrences
introduced in Section~\ref{sec:chebrec}. The question we now face is to compute
these solutions efficiently in spite of the various difficulties discussed above.
Our algorithm for this task may be viewed as a systematized variant of a
method originally due to Clenshaw~{\cite{Clenshaw1957}}. The link between
Clenshaw's method%
\footnote{The Clenshaw method we are referring to in this text should not be confused with the Horner-like scheme for Chebyshev polynomials known as Clenshaw's algorithm~\cite{Clenshaw1955}.}
and
the Chebyshev recurrence was observed long ago by Fox and
Parker~{\cite{Fox1962,FoxParker1968}} and further discussed by
Rebillard~{\cite[ Section~4.1.3]{Rebillard1998}}. Based on the properties of the
recurrence established in the last two sections, we can turn Clenshaw's method
into a true algorithm that applies uniformly to differential equations of
arbitrary order and degree.

Both Clenshaw's original method and our algorithm are related to Miller's
well-known {\emph{backward recurrence}}
technique~{\cite{BickleyComrieMillerSadlerThompson1952,Wimp1984}} to compute
minimal (``slowest increasing'') solutions of three-term recurrences. Miller's
idea is to compute the coefficients $u_N, u_{N - 1}, \ldots, u_0$ of a linear
recurrence sequence in the backward direction, starting form arbitrary
``initial conditions'' $u_{N + 1}$ and $u_{N + 2}$. When~$N$ goes to infinity
($u_{N + 1}, u_{N + 2}$ being chosen once and for all), the computed
coefficients $u_0, \ldots, u_N$ get close to those of a minimal solution with
large $u_0, u_1$, in accordance with the intuition that ``minimal solutions
are the dominant ones when going backwards''. This method behaves much better
numerically that the standard forward recurrence.
But its key feature for our purposes is that it allows one to compute a minimal solution characterized by its minimality plus one normalizing condition instead of two initial values.

Roughly speaking, our method amounts to a ``block Miller algorithm'' tuned to
the special case of Chebyshev recurrences. We use the idea of backward
recurrence to approximate the whole subspace of convergent solutions instead of a
single minimal one. There remains to take care of the constraints related to
the singularities of the recurrence, the symmetry condition and the initial values of the differential equation, all of which is done
using linear algebra.

\label{def:sing}
Denote $\mathbf{S}= \{ n \geqslant s : b_{- s} (n) = 0 \}$. Let
\[ \mathcal{E}= \{ (u_{| n |})_{n \in \mathbbm{Z}} : n \in
   \mathbbm{N}\backslash\mathbf{S} \Rightarrow (P \cdot u)_n = 0 \} \]
be the space of symmetric sequences whose restriction to $n \in \mathbbm{N}$
satisfies the Chebyshev recurrence, except possibly when $n \in \mathbf{S}$.

\begin{proposition}
  \label{prop:dimensions rec}
  \begin{enumerate}
    \item The space~$\mathcal{S}$ of symmetric sequences $(u_{| n
    |})_{n \in \mathbbm{Z}}$ such that $P \cdot u = 0$ has dimension $s + r$.
    Among the elements of~$\mathcal{E}$, these sequences are characterized by
    the linear equations
    \begin{equation}
      (P \cdot u)_n = 0, \hspace{2em} n \in \llbracket r, s - 1 \rrbracket
      \cup \mathbf{S} \label{eq:linear constraints} .
    \end{equation}
    These equations are linearly independent.
    
    \item The space $\mathcal{S} \cap \Cvgseq$ of symmetric, convergent
    sequences $(u_{| n |})_{n \in \mathbbm{Z}}$ such that $P \cdot u = 0$ has
    dimension~$r$. Its elements are the elements of~$\mathcal{E}$
    satisfying~\eqref{eq:linear constraints}, and the
    equations~\eqref{eq:linear constraints} are independent as linear forms
    on~$\mathcal{E} \cap \Cvgseq$ as well.
  \end{enumerate}
\end{proposition}

\begin{proof}
  First observe that a sequence~$u \in \mathcal{E}$ automatically satisfies $(
  P \cdot u)_n = 0$ for $- n \in \mathbbm{N}\backslash\mathbf{S}$ too, by
  Remark~\ref{rk:symmetry}. Then $u$ belongs to~$\mathcal{S}$ if and only if
  $(P \cdot u)_n = 0$ for $n \in \mathbf{S}$ and for $| n | < s$. But the
  equations $(P \cdot u)_n = 0$ with $| n | < r$ are trivial by
  Corollary~\ref{cor:symmetry}, thus $u \in \mathcal{S}$ if and only if $u \in
  \mathcal{E}$ and \eqref{eq:linear constraints}~is satisfied.
  
  Let $t = | \mathbf{S} |$, and let $p \leqslant s - r + t$~denote the rank
  of~\eqref{eq:linear constraints}, considered as a system of linear forms
  over~$\mathcal{E}$. For large enough~$n_0$, sequences $u \in \mathcal{E}$ are in bijection with their values
  $(u_n)_{n \in \mathbf{J}}$, where
  $\mathbf{J}= (\mathbf{S}- s) \cup \llbracket n_0, n_0 + 2 s - 1 \rrbracket$,
  hence $\mathcal{E}$~has dimension $2 s + t$. It follows that
  $\dim \mathcal{S}= \dim \mathcal{E}- p \geqslant s + r$.
  
  Similarly, an element of $\mathcal{E} \cap \Cvgseq$ is characterized,
  for large~$n_0$, by a convergent sequence $(u_n)_{n \geqslant n_0}$ (``a
  convergent germ of solution'') and values $u_n$ for $n \in \mathbf{S}$. It
  belongs to~$\mathcal{S}$ when additionally \eqref{eq:linear constraints}~is
  satisfied. Thus, by Corollary~\ref{cor:asympt sols}, we have $\dim (
  \mathcal{E} \cap \Cvgseq) = s + t$ and $\dim (\mathcal{S} \cap
  \Cvgseq) = s + t - q$, where $q \leqslant p$~is the rank of the
  system~\eqref{eq:linear constraints} restricted to~$\mathcal{E} \cap
  \Cvgseq$. But we already know from Theorem~\ref{thm:rec} that
  $\mathcal{S} \cap \Cvgseq \simeq \ker L$, and hence $\dim \mathcal{S}
  \cap \Cvgseq= r$. Therefore, we have $r = s + t - q \geqslant s + t - p
  \geqslant r$, and hence $p = q = s + t - r$.
\end{proof}

The important fact is that the equations~\eqref{eq:linear constraints} are
independent. The other statements are there to complete the picture but are
not really used in the sequel. Incidentally, Proposition~\ref{prop:dimensions
rec} implies that {\emph{divergent}} solutions of a Chebyshev recurrence have
the same exponential growth near positive and negative infinity.

The full procedure is stated as Algorithm~\ref{algo:Clenshaw}.
As the handling of boundary conditions is naturally incorporated into the algorithm, here, we see the initial conditions $y^{(i)}(0)=\ell_i$ as a special case of boundary conditions
\begin{equation}
  \label{eq:boundary}
  \lambda_i (y) = \ell_i, \quad 1 \leqslant i \leqslant r,
\end{equation}
each of the form $\lambda_i (y) = \sum_{j=1}^{q} \mu_j y^{(r_j)} (x_j)$ with
$x_j \in [- 1 , 1]$ and $r_j \leqslant r$.
In general, the boundary conditions are assumed to be chosen so that the function~$y$ of interest is the unique solution of~\eqref{eq:deq} satisfying~\eqref{eq:boundary}.
They are independent in the sense that the linear forms $\lambda_i : \ker L \to \mathbbm C$ are linearly independent. 

Motivated by
Proposition~\ref{prop:dimensions rec}, we ``unroll'' $s + t$ linearly
independent test sequences~$f_i \in \mathcal{E}$. We then solve the linear
system~\eqref{eq:sys Clenshaw}, consisting of the constraints~(\ref{eq:linear
constraints}) and of approximations of the boundary conditions~\eqref{eq:boundary}, to select a single linear combination of the~$f_i$ as
output. Algorithm~\ref{algo:Clenshaw} takes as input both a target degree~$d$
and a starting index~$N$. We will see in the next section how the choice
of~$N$ influences the quality of the output. In practice, taking $N = d + s$
usually yields good results.

\begin{algorithm} \label{algo:Clenshaw}

\emph{Input:}
  a linear differential operator~$L$ of order~$r$,
  boundary conditions ${\lambda_1(y)=\ell_1}, \ldots, {\lambda_r(y)=\ell_r}$
  as in~(\ref{eq:boundary}),
  a target degree $d > s$,
  an integer $N \geq \max(d, \max \{n : b_{-s}(n)=0 \})$.
\emph{Output:}
  an approximation $\tilde{y} (x) = \sum_{n = - d}^d \tilde{y}_n T_n (x)$
  of the corresponding solution~$y$ of $L \cdot y = 0$.

\newcounter{algoline}
\begin{body}
  \item compute the Chebyshev recurrence operator $P = \sum_{k = - s}^s b_k (n) S^k$
  associated to~$L$
  
  \item set $\mathbf{S}= \{n \geq s : b_{- s} (n) = 0\}$ and
  $\mathbf{I}=\mathbf{S} \cup \llbracket N, N + s - 1 \rrbracket$

  \labelitem{step:unroll}
  for $n$ from $N+2s-1$ down to $1$
    \begin{body}
    \item for $i \in \mathbf{I}$

      \begin{body}
      \item if $n=i$ then set $f_{i,n-s} = 1$
      \item else if $n \in \mathbf I$ or $n \geq N+s$ then set $f_{i,n-s}=0$
      \item else compute $f_{i,n-s}$ using the relation $(P \cdot f)_n=0$
      \end{body}
    \end{body}

\item using indeterminates~$\eta_i$, ${i \in \mathbf I}$, set
  \[
    \tilde{y}_n = 
    \left\{ \begin{array}{ll}
      \sum_{i \in \mathbf{I}} \eta_i f_{i, \left| n \right|}, &
      \left| n \right| \leqslant N \\
      \tilde{y}_n = 0, &
      \left| n \right| > N,
    \end{array} \right.
    \quad \text{and} \quad
    \tilde y(x) = \sum_{n=-N}^{N} \tilde y_n T_n(x)
  \]
  
  \labelitem{step:sys Clenshaw}
  solve for $\left( \eta_i \right)_{i \in \mathbf{I}}$ the linear system
  \begin{equation} \label{eq:sys Clenshaw}
    \left\{ \begin{array}{ll}
      \lambda_k ( \tilde{y}) = \ell_k, &
        1 \leqslant k \leqslant r,\\
      b_{- s} (n) \tilde{y}_{n - s} + \cdots + b_s (n) \tilde{y}_{n + s} = 0,
      &  n \in \llbracket r, s -1  \rrbracket \cup \mathbf{S}
    \end{array} \right. 
  \end{equation}

  \item return $\sum_{n = - d}^d \tilde{y}_n T_n (x)$
\end{body}

\end{algorithm}

The complexity of Algorithm~\ref{algo:Clenshaw} is easy to estimate.

\begin{proposition}
  \label{prop:complexity Clenshaw}For fixed~$L$, $\lambda_i$, and $\ell_i$,
  Algorithm~\ref{algo:Clenshaw} runs in $O (N)$ arithmetic operations.
\end{proposition}

Its correctness is less obvious. At first sight, there could conceivably exist
differential equations for which Algorithm~\ref{algo:Clenshaw} always fails,
no matter how large~$N$ is chosen. It could happen for instance that the
kernel of~\eqref{eq:linear constraints} (a system we know to have full rank
over~$\mathcal{E}$ by Proposition~\ref{prop:dimensions rec}) always has a
nontrivial intersection with $\operatorname{Span} \{ t_i : i \in \mathbf{I} \}$. It
is not clear either that, when the algorithm does return a polynomial~$p$,
this polynomial is close to~$y$. We prove in the next section that these issues do not occur. But already at this point, we note that if the result happens
to be satisfactory, it is already possible to validate it (that is, to get a
rigorous \emph{good} upper-bound on $\| y - p \|_{\infty}$)
using the methods of Section~\ref{sec:validation}.

\subsection{Convergence}

We now prove that Algorithm~\ref{algo:Clenshaw} converges. The fact that it
does not fail for large~$N$ will come as a byproduct of the convergence
proof. The proof, inspired in part by the analysis of the generalized Miller
algorithm~\cite{Zahar1976,Wimp1984}, is based on the asymptotic
behaviour of the solutions of the Chebyshev recurrence predicted by
Theorem~\ref{thm:BT}.
The approach of backward recurrence algorithms based on this theorem was pioneered by Wimp~{\cite{Wimp1969}}.

Retaining the notation from the previous subsection, assume that the
operator~$L$ and the boundary conditions $\lambda_i (y) = \ell_i$ are fixed.
Write the Chebyshev expansion of~$y$ as
\[ y (x) = \sum_{n = - \infty}^{\infty} y_n T_n (x) . \]
Let $y^{(N)}_n = \tilde{y}_n$, $\mathopen{|} n \mathclose{|} \leqslant N$, be
the coefficients computed by Algorithm~\ref{algo:Clenshaw} (run in exact
arithmetic) when called with the starting index~$N$.
\label{def:yN}

The central result of the analysis of Algorithm~\ref{algo:Clenshaw} is the
following theorem. It implies that when~$d$ is fixed and $N \rightarrow
\infty$, the polynomial output by Algorithm~\ref{algo:Clenshaw} converges at
least exponentially fast to the truncated Chebyshev series~$\pi_d (y)$. The
base of the exponential is related to the asymptotics of the ``slowest
decreasing'' convergent solution of the Chebyshev recurrence, in turn related
to the location of the singular points of the differential
equation~\eqref{eq:deq}.

\begin{theorem}
  \label{thm:convergence}Algorithm~\ref{algo:Clenshaw} fails for finitely
  many~$N$ only. As $N \to \infty$, its output satisfies
  \[ \max_{n = - N}^N |y_n^{(N)} - y_n | = O (N^{\tau} e_{1, N}) \]
  for some~$\tau$ independent of~$N$.
\end{theorem}

We write $f (N) = O_{\operatorname{pol}} (g (N))$ when there exists $\tau \geqslant
0$ such that $f (N) = O (N^{\tau} g (N))$ as $n \rightarrow \infty$.

\begin{proof}
  The finite sequence $(y^{(N)}_n)_{n = - N}^N$ computed by
  Algorithm~\ref{algo:Clenshaw} extends to an element $(y^{(N)}_n)_{n \in
  \mathbbm{Z}}$ of~$\mathcal{E}$ characterized by the conditions
  \[ y^{(N)}_N = \cdots = y^{(N)}_{N + s - 1} = 0 \]
  from Step~\ref{step:unroll}, along with the linear system~(\ref{eq:sys
  Clenshaw}) solved in Step~\ref{step:sys Clenshaw}.
  
  By writing the linear forms $\lambda_1, \ldots, \lambda_r : \Cvgseq
  \rightarrow \mathbbm{C}$ that express the boundary
  conditions~\eqref{eq:boundary} as $\lambda_i (y) = \sum_{n = -
  \infty}^{\infty} \lambda_{i, n} y_n$, we define ``truncations''
  \[ \lambda_i^{(N)} (y) = \sum_{n = - N}^N \lambda_{i, n} y_n \]
  that make sense even for divergent series. (Abusing notation slightly, we
  apply the~$\lambda_i$ and~$\lambda_i^{(N)}$ indifferently to functions,
  formal Chebyshev series or their coefficient sequences.) The
  system~\eqref{eq:sys Clenshaw} consists of the equations $\lambda^{(N)}_i
  (y^{(N)}) = \ell_i$ and of the symmetry and extension-through-singularities
  constraints~\eqref{eq:linear constraints}. We introduce additional linear
  forms $\lambda_{r + 1} = \lambda_{r + 1}^{(N)}, \ldots, \lambda_{s + t} =
  \lambda_{s + t}^{(N)}$ to write these last~$s - r + t$ equations in the
  same form as the first~$r$, so that \eqref{eq:sys Clenshaw}~becomes
  \begin{equation}
    \lambda_i^{(N)} (y^{(N)}) = \sum_{n = - N}^N \lambda_{i, n} y^{(N)}_n =
    \ell_i, \hspace{1em} 1 \leqslant i \leqslant s + t \text{\label{eq:exact
    linear constraints}} .
  \end{equation}

  Now let $(e_1, \ldots, e_s, e_{- 1}, \ldots, e_{- s})$ be a basis of the
  solutions of~$P \cdot u = 0$ in the neighborhood of~$+ \infty$ of the form
  provided by Corollary~\ref{cor:asympt sols}. Extend each $e_i$ to an element
  of~$\mathcal{E}$, and then the tuple to a basis of~$\mathcal{E}$ by setting
  $e_{s + 1} = f_{n_1}, \ldots, e_{s + t} = f_{n_t}$ where $n_1 < n_2 < \cdots
  < n_t$ are the elements of~$\mathbf{S}$. Let
  \begin{equation}
    \Delta^{(N)} = \left|\begin{array}{cccccc}
      e_{1, N} & \cdots & e_{s + t, N} & e_{- 1, N} & \cdots & e_{- s, N}\\
      \vdots &  & \vdots & \vdots &  & \vdots\\
      e_{1, N + s - 1} & \cdots & e_{s + t, N + s - 1} & e_{- 1, N + s - 1} &
      \cdots & e_{- s, N + s - 1}\\
      \lambda_1^{(N)} (e_1) & \cdots & \lambda_1^{(N)} (e_{s + t}) &
      \lambda_1^{(N)} (e_{- 1}) & \cdots & \lambda_1^{(N)} (e_{- s})\\
      \vdots &  & \vdots & \vdots &  & \vdots\\
      \lambda_{s + t}^{(N)} (e_1) & \cdots & \lambda_{s + t}^{(N)} (e_{s +
      t}) & \lambda_{s + t}^{(N)} (e_{- 1}) & \cdots & \lambda_{s + t}^{(
      N)} (e_{- s})
    \end{array}\right|, \label{eq:large det}
  \end{equation}
  Let $\Delta_j^{(N)}$ be the same determinant with the column
  involving~$e_j$ replaced by
  \[ (\underbrace{0, \ldots, 0}_{\text{$s$ times}}, \ell_1, \ldots,
  \ell_r, \underbrace{0, \ldots, 0}_{\text{$s - r + t$ \rlap{times}}})^{\Tau} . \]
  By Cramer's rule, provided $\Delta^{(N)} \neq 0$, the sequence $(y^{(
  N)}_n)_n$ decomposes on the basis~$(e_j)_{j = - s}^{s+t}$ of $\ker P \subset
  \mathbbm{C}^{\mathbbm{Z}}$ as
  \begin{equation}
    y^{(N)} = \sum_{k = - s}^{s + t} \gamma_k^{(N)} e_k, \hspace{2em}
    \gamma_k^{(N)} = \frac{\Delta_k^{(N)}}{\Delta_{}^{(N)}} .
    \label{eq:decomp approx}
  \end{equation}
  Algorithm~\ref{algo:Clenshaw} fails if and only if $\Delta^{(N)} = 0$.
  
  The sequence of ``exact'' Chebyshev coefficients of the function~$y$ defined
  by the input is likewise given by
  \begin{equation}
    y = \sum_{k = 1}^{s + t} \gamma_k e_k, \hspace{2em} \gamma_k =
    \frac{\Delta_k}{\Delta}, \label{eq:decomp exacte}
  \end{equation}
  where
  \[ \Delta = \det (\lambda_i (e_j))_{1 \leqslant i, j \leqslant s + t} \]
  and~$\Delta_j$ denotes the determinant~$\Delta$ with the $j$-th column
  replaced by $(\ell_1, \ldots, \ell_{s + t})^{\Tau}$.
  
  Our goal is now to prove that $\gamma_k^{(N)} \rightarrow \gamma_k$ fast as
  $N \rightarrow \infty$. To do that, we study the asymptotic behaviours of
  the determinants $\Delta^{(N)}$ and $\Delta^{(N)}_k$.
  
  We decompose~$\Delta^{(N)}$ into the four blocks indicated by
  Eq.~\eqref{eq:large det} as follows:
  \[ \Delta^{(N)} = \left|\begin{array}{cc}
       A & B\\
       C & D^{}
     \end{array}\right| . \]
  The corresponding modified blocks in~$\Delta^{(N)}_k$ are denoted $A_k$,
  $B^{}_k$, $C^{}_k$, $D_k$. (We drop the explicit index for readability,
  but notice that these matrices depend on~$N$.) The blocks $B$~and~$C$ are
  nonsingular for large~$N$, the first one by Lemma~\ref{lemma:Casorati} and
  the second one because $\det C \rightarrow \Delta \neq 0$ as~$N \rightarrow
  \infty$. The Schur complement formula implies
  \[ \Delta^{(N)} = - \det (B) \det (C) \det (I - C^{- 1} DB^{- 1} A) . \]
  Setting $\mathbf{e}_j = (e_{j, N}, \ldots, e_{j, N + s - 1})^{\Tau}$, the
  entry at position~$(i, j)$ in the matrix~$B^{- 1} A$ satisfies $(B^{- 1}
  A)_{i, j} = 0$ for large~$N$ if $j > s$, and otherwise
  \begin{align*}
    (B^{- 1} A)_{i, j} &= \frac{\det (\mathbf{e}_{- 1}, \ldots,
    \mathbf{e}_{- i + 1}, \mathbf{e}_j, \mathbf{e}_{- i - 1}, \ldots,
    \mathbf{e}_{- s})}{\det B}\\
    &= \frac{(- 1)^{i - 1} \det (\mathbf{e}_j, \mathbf{e}_{- 1},
    \ldots, \widehat{\mathbf{e}_{- i}},, \ldots, \mathbf{e}_{- s})}{\det
    B}\\
    &= O_{\operatorname{pol}} \left(\frac{e_{j, N}}{e_{- i, N}} \right)
  \end{align*}
  (where the notation $\widehat{\cdot}$ indicates the omission of the corresponding term)
  as $N \rightarrow \infty$ by Lemma~\ref{lemma:Casorati}.
  
  In view of our assumptions on the boundary conditions~\eqref{eq:boundary},
  we have $\lambda_{i, n} = O_{n \rightarrow \pm \infty} (n^r)$ for all $i
  \leqslant \sigma$, where $\sigma$~is the maximum derivation order appearing in~\eqref{eq:boundary}.
  Additionally, the sequences $\lambda_{i, n}$ with $r + 1
  \leqslant i \leqslant s + t$ are ultimately zero. Therefore the entries
  of~$D$ satisfy 
  \[ D_{i, j} = \lambda^{(N)}_i (e_{- j}) = O_{\operatorname{pol}} (
  e_{- j, N}). \]
  This yields the estimate
  \[ (DB^{- 1} A)_{i, j} = O_{\operatorname{pol}} (e_{j, N}) \]
  for the $j$-th column of $DB^{- 1} A$. Since
  \[ C_{i, j} = \lambda_i^{(N)} (e_j) = \lambda_i (e_j) + O_{\operatorname{pol}} (
     e_{j, N}), \]
  we get $(C^{- 1} DB^{- 1} A)_{i, j} = O_{\operatorname{pol}} (e_{j, N})$ as well,
  and
  \begin{align*}
    \Delta^{(N)} &= - \det (B) \det (C)  (1 - \operatorname{tr} (C^{- 1} DB^{-
    1} A) + O (\| C^{- 1} DB^{- 1} A \|^2))\\
    &= - \det (B)  (\Delta + O_{\operatorname{pol}} (e_{1, N})) .
  \end{align*}
  In particular, $\Delta^{(N)} \neq 0$ for all large enough~$N$, hence, for
  any fixed differential equation, the algorithm fails at most for finitely
  many~$N$.
  
  We turn to the modified determinants
  \[ \Delta_k^{(N)} = \left|\begin{array}{cc}
       A_k & B_k\\
       C_k & D_k
     \end{array}\right| . \]
  For $k > 0$, the same reasoning as above (except that~$C_k$ may now be
  singular) leads to
  \begin{align*}
    \Delta^{(N)}_k &= - \det (B) \det (C_k - DB^{- 1} A_k)\\
    &= - \det (B)  (\det (C_k) + O_{\operatorname{pol}} (e_{1, N}))\\
    &= - \det (B)  (\Delta_k + O_{\operatorname{pol}} (e_{1, N})),
  \end{align*}
  hence
  \begin{equation}
    \gamma_k^{(N)} = \frac{\Delta_k^{(N)}}{\Delta^{(N)}} = \gamma_k +
    O_{\operatorname{pol}} (e_{1, N}), \hspace{2em} k > 0. \label{eq:estimation
    approx cvgte}
  \end{equation}

  In the case $k < 0$, write
  \[ \Delta_k^{(N)} = - \det (C) \det (B_k - AC^{- 1} D_k) . \]
  The natural entrywise bounds on~$A$ and~$D$ yield $(C^{- 1} D_k)_{i, j} = O
  (N^r e_{- j, N + s - 1})$ and from there
  \[ (AC^{- 1} D_k)_{i, j} = O (N^r e_{1, N} e_{- j, N + s - 1}) = o (e_{-
     j, N}), \]
  so that
  \[ (B_k - AC^{- 1} D_k)_{i, j} \sim e_{- j, N + i - 1}, \hspace{2em} j \neq
     - k. \]
  For $j = - k$ however, the $j$-th column of~$B_k$ is zero and that of~$D_k$
  is constant, hence
  \[ (B_k - AC^{- 1} D_k)_{i, j} = O (e_{1, N}), \hspace{2em} j = - k. \]
  It follows that
  \[ \det (B_k + AC^{- 1} D_k) = O_{\operatorname{pol}} (e_{- 1, N} \cdots
     \widehat{e_{k, N}} \cdots e_{- s, N} e_{1, N}) = O_{\operatorname{pol}} \left(
     \frac{\det (B)}{e_{k, N}} e_{1, N} \right), \]
  whence
  \begin{equation}
    \gamma_k^{(N)} = \frac{\Delta_k^{(N)}}{\Delta^{(N)}} = \frac{- \det (
    B) \det (C) O_{\operatorname{pol}} (e_{1, N} / e_{k, N})}{- \det (B) \det (C) 
    (1 + O (e_{1, N}))} = O_{\operatorname{pol}} \left(\frac{e_{1, N}}{e_{k, N}}
    \right), \hspace{1em} k < 0. \label{eq:estimation approx dvgte}
  \end{equation}

  Combining~\eqref{eq:decomp approx}, \eqref{eq:decomp exacte}
  with~\eqref{eq:estimation approx cvgte}, \eqref{eq:estimation approx dvgte}
  finally yields
  \[ y_n^{(N)} = y_n + O_{\operatorname{pol}} \left(e_{1, N} \sum_{k = 1}^s \left(
     e_{k, n} + \frac{e_{- k, n}}{e_{- k, N}} \right) \right) \]
  as $N \rightarrow \infty$, uniformly in~$n$.
\end{proof}

\begin{remark}
  \label{rk:exact}
  In the special case where the solution~$y$ is a polynomial, it is computed
  exactly.
\end{remark}

Theorem~\ref{thm:convergence} implies that the polynomial $y^{(N)}=\tilde y$ returned by Algorithm~\ref{algo:Clenshaw} satisfies
\begin{equation}
  \label{eq:err of N}
  \| y^{(N)} - \pi_d(y) \|_{\infty}
  \leq
  \phi(N) N!^{\kappa_1} \alpha_1^N,
\end{equation}
where $\kappa_1$~and~$\alpha_1$ are the asymptotic growth parameters defined in Section~\ref{sec:Newton polygon}, for some~$\phi$ with $\ln \phi(N) = o(N)$.
Thus, given $\epsilon > 0$, it suffices to take
$N=O(\ln(\epsilon^{-1}))$
in order to obtain
$\| y^{(N)} - \pi_d(y) \|_{\infty} \leq \epsilon$.
The constant hidden in the~$O(\cdot)$ depends on~$y$.
The estimate goes down to
$O\bigl(\ln(\epsilon^{-1})/\ln \ln(\epsilon^{-1})\bigr)$
when~$\kappa_1 < 0$, that is (by a similar argument as in the proof of Proposition~\ref{prop:Newton polygon}), when the leading coefficient~$a_r$ of the differential equation is a constant.

Comparing with Equation~\eqref{eq:Lebesgue}, we can state the following ``effective near-minimax approximation'' property.

\begin{corollary} \label{cor:effective near-minimax}
  Let $L$ and $(\ell_k)_{k=1}^r$ be fixed.
  Given~$d \in \mathbb N$, there exists~$N$ such that Algorithm~\ref{algo:Clenshaw}, called with parameters $L$, $(\ell_k)$, $d$, and~$N$, computes a polynomial~$p_d$ of degree at most~$d$ satisfying
  $\| p_d - y \|_{\infty} \leq \bigl(4 \pi^{-2} \ln(d+1) + 5 \bigr) \, \|p^{\ast}_d - y\|_{\infty}$
  in $O(\ln \| p_d - y \|_{\infty}^{-1})$ arithmetic operations.
\end{corollary}

There is a different way of looking at this, starting with the lower bound~\cite[Sec.~4.4, Theorem~5(i)]{Cheney1998}
\begin{equation}
  \label{eq:lower bound minimax}
  \| p^{\ast}_d - y \|_{\infty} \geq \frac\pi2 \max_{n>d} |y_n|
\end{equation}
on the quality of the minimax polynomial approximation of degree~$d$ of a function~$y$ in terms of the Chebyshev coefficients~$y_n$ of~$y$.
In the (typical) case where
\begin{equation}
  \label{eq:hyp rk num terms}
  \max_{k=0}^{r-1} |y_{n+k}| \geq n!^{\kappa_1} |\alpha_1|^n \psi(n),
  \qquad
  \psi(n)= e^{o(n)},
\end{equation}
we see by comparing with~\eqref{eq:err of N} that choosing $N = d + o(d)$ is enough to get
\[
  \| y^{(N)} - \pi_d(y) \|_{\infty}
  \leq \| p^{\ast}_d - y \|_{\infty}.
\]
This last inequality in turn implies
\[
  \| y^{(N)} - y \|_{\infty}
  \leq
  \bigl(4 \pi^{-2} \ln(d+1) + 5 \bigr) \, \|p^{\ast}_d - y\|_{\infty}.
\]
When, in~\eqref{eq:hyp rk num terms}, $\kappa_1$ and $\alpha_1$ are replaced by $\kappa_i$ and $\alpha_i$ for some $i>1$, the estimate $N=d+o(d)$ still holds in the case where $\kappa_i = \kappa_1 < 0$.
It becomes $N=O(d)$ when either $\kappa_i = \kappa_1 = 0$ or $\kappa_i < \kappa_1 < 0$, and $N = O(d \ln d)$ in general (that is, when $\kappa_i < \kappa_1 = 0$).

\begin{remark}
  Assuming only $e_{i, n + 1} / e_{i, n} \sim \alpha_i n^{\kappa_i} e_{i, n}$
  instead of full asymptotic expansions in Lemma~\ref{lemma:Casorati}, but
  with the additional hypothesis $\kappa_i = \kappa_j \Rightarrow \alpha_i
  \neq \alpha_j$, one can prove that
  \[ C (n) \sim_{n \to \infty} e_{0, n} e_{1, n + 1} \cdots e_{s - 1, n + s -
    1}  \prod_{\substack{
       i < j\\
       \kappa_i = \kappa_j
 }} \left(\frac{\alpha_i}{\alpha_j} - 1 \right), \]
  if the $e_i$ are sorted so that $\kappa_0 \leqslant \kappa_1 \leqslant
  \cdots \leqslant \kappa_{s - 1}$. This leads to a weaker variant of
  Theorem~\ref{thm:convergence} that does not rely on the Birkhoff-Trjitzinsky
  theorem.
\end{remark}

\subsection{Examples}\label{sec:examples-clenshaw}

\begin{figure}
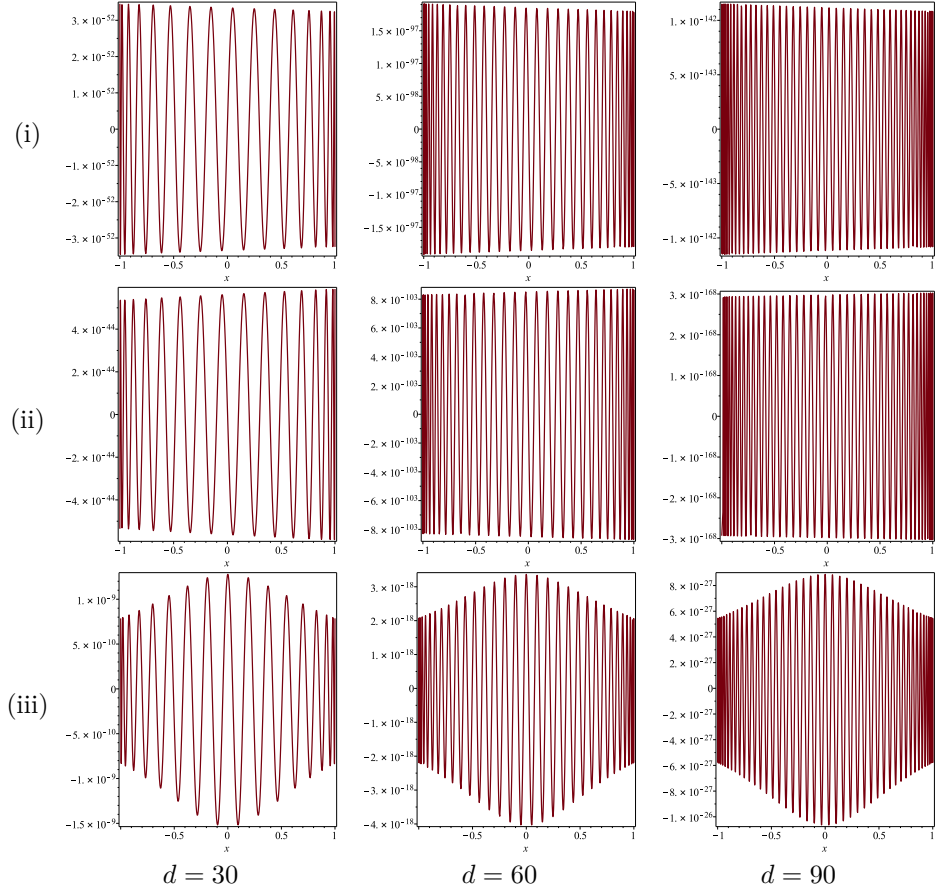

  \newcommand{\includeexample}[2]{
    \raisebox{-.5\height}{%
      \includegraphics[width=.30\textwidth,trim=5pt 5pt 5pt 5pt ,clip]{Figures/example#1Error#2}
    }%
  }
  \renewcommand{\tabcolsep}{.2ex}
  \centering
  \begin{tabular}{cccc}
    (\ref{ex:ex1AMM})
    & \includeexample{1}{30}
    & \includeexample{1}{60}
    & \includeexample{1}{90}
    \\
    (\ref{ex:ex2AMM})
    & \includeexample{2}{30}
    & \includeexample{2}{60}
    & \includeexample{2}{90}
    \\
    (\ref{ex:ex3AMM})
    & \includeexample{3}{30}
    & \includeexample{3}{60}
    & \includeexample{3}{90}
    \\
    \noalign{\smallskip}
    & $d = 30$ & $d = 60$ & $d = 90$ \\
  \end{tabular}
  \caption{Plot of the error $p_d-y$ between a degree-$d$ approximation~$p_d$ computed by Algorithm~\ref{algo:Clenshaw} and the exact solution$y$, for each of the problems listed in Section~\ref{sec:examples-clenshaw} and for $d \in \{30, 60, 90\}$.}
  \label{fig:plots}
\end{figure}

We have developed a prototype implementation of Algorithm~\ref{algo:Clenshaw} in Maple~\cite{Maple}.
Our implementation uses exact rational arithmetic for operations on the coefficients of approximation polynomials.
The experimental source code can be downloaded from
\begin{center}
  \url{http://homepages.laas.fr/mmjoldes/Unifapprox}
\end{center}
Besides Algorithm~\ref{algo:Clenshaw}, it includes a (not entirely rigorous with respect to several minor outwards rounding issues) proof-of-concept implementation of the validation algorithm of Section~\ref{sec:validation} further discussed in that section.
The \texttt{gfsRecurrence} package~\cite{Benoit2012} available on the same web page provides tools to compute Chebyshev recurrences from linear differential equations as discussed in Section~\ref{sec:chebrec}.

For each of the following examples, Figure~\ref{fig:plots} shows the graph of the difference between the polynomial approximation of a given degree computed by the implementation and the known exact solution, illustrating the quality of the approximations.
\begin{enumerate}
\item \label{ex:ex1AMM}
The first example is adapted from Kaucher and Miranker~\cite[p.~222]{KaucherMiranker1984}.
It concerns the hyperexponential function
\[ y(x)=\dfrac{e^{x/2}}{\sqrt{x+16}}, \]
which can be defined by the differential equation
\[ 2\,(x+16)y'(x) - (x+15)y(x) = 0, \quad  y(0) = \frac14. \]

\item\label{ex:ex2AMM}
Next, we consider the fourth order initial value problem
(taken from Geddes~\cite[p.~31]{Geddes1977a})
\[
  y^{(4)}(x)-y(x)=0, \quad
  y(0) = -y''(0) = \frac32, \quad
  -y'(0) = y'''(0) = \frac12,
\]
with the exact solution
\[
y(x) = \frac32 \cos(x) - \frac12 \sin(x).
\]

\item\label{ex:ex3AMM}
Finally, the second-order differential equation
\[ (2x^2+1) y''(x)+8x y'(x) + (2x^2+5) y(x), \quad y(0) =1, y'(0) =0, \]
has complex singular points at $z=\pm i/\sqrt{2}$,
relatively close to the interval~$[-1,1]$,
and admits the exact solution
\[
y(x)=\frac {\cos(x) }{2x^2+1}.
\]
\end{enumerate}

On our test system, using Maple~17, the total computation time for each example is of the order of 0.05 to 0.1~s.

According to a classical theorem of de la Vallée Poussin~\cite[Section~3.4]{Cheney1998}, the near-uniform amplitude of the oscillation observed in the first two examples indicates an approximation error very close to that of the minimax approximation.
Table~\ref{tab:resultsAMM} in Section~\ref{sec:validation-algorithm} (p.~\pageref{tab:resultsAMM}) gives numerical values of $\|p-y\|_{\infty}$ and $\|p^{\ast}-y\|_{\infty}$ in each case.
We will later extend these examples to include in the comparison the bounds on $\|p-y\|_{\infty}$ output by the validation algorithm.

\subsection{A Link with the Tau Method}

Besides Clenshaw's, another popular method for the approximate computation of Chebyshev expansions
is Lánczos' tau method~{\cite{Lanczos1938,Lanczos1956}}. It has been
observed by Fox~\cite{Fox1962} and later in greater generality (and
different language) by El Daou, Ortiz and
Samara~{\cite{ElDaouOrtizSamara1993}} that both methods are in fact
equivalent, in the sense that they may be cast into a common framework and
tweaked to give exactly the same result. We now outline how the use of the
Chebyshev recurrence fits into the picture. This sheds another light on
Algorithm~\ref{algo:Clenshaw} and indicates how the Chebyshev recurrence may
be used in the context of the tau method.

As in the previous sections, consider a differential equation~$L \cdot y = 0$
of order~$r$, with polynomial coefficients, to some solution of which a
polynomial approximation of degree~$d$ is sought. Assume for simplicity that
there are no nontrivial polynomial solutions, i.e., $(\ker L)
\cap \mathbbm{C} [x] = \{ 0 \}$.

In a nutshell, the tau method works as follows. The first step is to
compute~$L \cdot p$ where~$p$ is a polynomial of degree~$d$ with indeterminate
coefficients. Since $(\ker L) \cap \mathbbm{C} [x] =
\{ 0 \}$, the result has degree greater than~$d$. One then
introduces additional unknowns $\tau_{d + 1}, \ldots, \tau_{d + m}$ in such
number that the system
\begin{equation}
\label{eq:sys tau}
  \left\{
  \begin{aligned}
    L \cdot p &  = \tau_{d + 1} T_{d + 1} + \cdots + \tau_{d + m} T_{d + m} \\
    \lambda_i (p) &  = \ell_i & (1 \leqslant i \leqslant r)
  \end{aligned}
  \right.
\end{equation}
has a (preferably unique) solution. The output is the value of~$p$ obtained by
solving this system; it is an exact solution of the projection~$\pi_d (L
\cdot y) = 0$ of the original differential equation.

Now let $p = \sum_{n = - d}^d p_n T_n$ and extend the sequence~$(\tau_n)$ by
putting~$\tau_n = 0$ for~$\abs n \not\in \llbracket d + 1, d + m
\rrbracket$ and $\tau_{-n} = \tau_n$. It follows from~(\ref{eq:sys tau})
that~$P \cdot (p_n) = \frac{1}{2} Q \cdot (\tau_n)$
where~$P$ and~$Q$ are the recurrence operators given by Theorem~\ref{thm:rec}.
Denoting $\operatorname{Supp} u = \{ \abs n : u_n \neq 0 \}$, we also see from the
explicit expression of~$Q$ that ${\operatorname{Supp}} (Q \cdot \tau) \subset
\llbracket d, d + m + 1 \rrbracket$. Hence the coefficients~$p_n$ of the
result of the tau method are given by the Chebyshev recurrence, starting from
a small number of initial conditions given near the index~$\mathopen| n \mathclose| = d$.

``Conversely,'' consider the polynomial~$\tilde{y}$ computed in
Algorithm~\ref{algo:Clenshaw} with $N=d$, and let $v = \sum_n v_n T_n = L \cdot
\tilde{y}$. We have $P \cdot \tilde{y} = Q \cdot v$ by Theorem~\ref{thm:rec}.
But the definition of~$\tilde{y}$ in the algorithm also implies that~$(P
\cdot \tilde{y})_n = 0$ when $\mathopen| n \mathclose| \leqslant N - s$ (since
the~$\tilde{y}_n$, $\mathopen| n \mathclose| \leqslant N$ are linear
combinations of sequences~$(f_{i, n})_{\mathopen| n \mathclose| \leqslant N}$
recursively computed using the recurrence~$P \cdot f_i = 0$) or $\mathopen| n
\mathclose| > N + s$
(since~$\tilde{y}_n = 0$ for $\mathopen| n \mathclose| > N$), so that
${\operatorname{Supp}} (Q \cdot v) \subset \llbracket
N - s, N + s - 1 \rrbracket$. It can be checked that the Chebyshev recurrence
associated to~$L = (\frac{\mathrm{d}}{\mathrm{d} x})^r$ is~$P = \delta_r (n)$ : indeed, in the
language of~{\cite{BenoitSalvy2009}}, it must be the first element of a pair $(P_1,Q_1)$ satisfying $Q_1^{-1} P_1 = I^{-r}$.
Thus $\delta_r (n) \cdot u = Q \cdot v$ is equivalent to
$u^{(r)} = v$, whence
\begin{equation}
  v (x) = \frac{\mathrm{d}^r}{\mathrm{d} x^r}  \sum_{\mathopen| n
  \mathclose| > r} \frac{(P \cdot \tilde{y})_n}{\delta_r (n
)} T_n (x) = \sum_{N - s \leqslant \mathopen| n \mathclose| < N +
  s} \frac{(P \cdot \tilde{y})_n}{\delta_r (n)}
  T^{(r)}_n (x). \label{eq:rhs Clenshaw tau}
\end{equation}
We see that the output~$\tilde{y} (x)$ of Algorithm~\ref{algo:Clenshaw}
satisfies an inhomogeneous differential equation of the form~$L \cdot
\tilde{y} = \tau_{N - s} T_{N-s}^{(r)} (x) + \cdots +
\tau_{N + s - 1} T_{N+s-1}^{(r)} (x)$. (However, the
support of the sequence~$(v_n)$ itself is not sparse in general.)

This point of view also leads us to the following observation.

\begin{proposition}
  Assume that Equation~\ref{eq:deq} has no polynomial solution.
  The expression \emph{on the monomial basis} of the polynomial~$\tilde{y}(x)$ returned by Algorithm~\ref{algo:Clenshaw} with $N=d$ can be computed in~$O(d)$ arithmetic operations, all other parameters being fixed.
\end{proposition}

 In comparison,
the best known arithmetic complexity bound for the conversion of arbitrary
polynomials of degree~$d$ from the Chebyshev basis to the monomial basis is~$O(M(d))$, where~$M$ stands for the cost of
polynomial multiplication~{\cite{Pan1998, BostanSalvySchost2008}}.

\begin{proof}
As already mentioned, the Taylor series expansion of
a function that satisfies an LODE with polynomial coefficients obeys a linear
recurrence relation with polynomial coefficients. In the case of an
inhomogeneous equation~$L \cdot u = v$, the recurrence operator does not
depend on~$v$, and the right-hand side of the recurrence is the coefficient
sequence of~$v$. Now~$\tilde{y}$ satisfies~$L \cdot \tilde{y} = v$ where~$v$
is given by~(\ref{eq:rhs Clenshaw tau}). The coefficients $(P \cdot
\tilde{y})_n / \delta_r (n)$ of~(\ref{eq:rhs Clenshaw
tau}) are easy to compute from the last few Chebyshev coefficients
of~$\tilde{y}$.
One deduces the coefficients~$v_n$ in linear time by applying repeatedly the non-homogeneous recurrence relation
\begin{equation} T_{n-1}^{\prime}(x) = -T_{n+1}^{\prime}(x)+2xT_n^{\prime}(x)+2T_n(x) \end{equation} obtained by differentiation of the equation~(\ref{eq:recTn}), and
finally those of the expansion of~$\tilde y$ on the monomial basis using the
recurrence relation they satisfy.
\end{proof}

\section{Chebyshev Expansions of Rational Functions}\label{sec:ratpolys}

This section is devoted to the same problems as the rest of the article,
only restricted to the case where~$y (x)$ is a rational function. We are
interested in computing a recurrence relation on the coefficients~$y_n$ of the
Chebyshev expansion of a function~$y$, using this recurrence to obtain a good
uniform polynomial approximation of~$y (x)$ on~$[- 1, 1]$, and certifying the
accuracy of this approximation. All this will be useful in the validation part
of our main algorithm.

Our primary tool is the change of variable $x = \frac{1}{2}  (z + z^{- 1})$
followed by partial fraction decomposition. Similar ideas have been used in
the past with goals only slightly different from ours, like the computation
of~$y_n$ in closed form~{\cite{EinwohnerFateman1989,Mathar2006}}. Indeed, the
sequence~$(y_n)_{n \in \mathbbm{N}}$ turns out to obey a recurrence with
{\itshape{constant}} coefficients. Finding this recurrence or a closed form
of~$y_n$ are essentially equivalent problems. However, we will use results
regarding the cost of the algorithms that do not seem to appear in the
literature. Our main concern in this respect is to avoid conversions of
polynomial and series from the monomial to the Chebyshev basis and back. We
also need simple error bounds on the approximation of a rational function by its Chebyshev expansion.

\subsection{Recurrence and Explicit Expression}\label{sec:chebcoeff ratpoly}

Let
$y (x) = a (x) / b (x) \in \mathbbm{Q} [x]$
be a rational function with no pole in~$[- 1, 1]$.
\label{def:y rat}
As usual, we denote by~$(y_n)_{n \in \mathbbm{Z}}$,
$(a_n)_{n \in \mathbbm{Z}}$ and $(b_n)_{n \in \mathbbm{Z}}$
the symmetric Chebyshev coefficient sequences of~$y$, $a$ and~$b$.

\begin{proposition}
  \label{prop:chebrec ratpoly}The Chebyshev coefficient sequence~$(y_n
)_{n \in \mathbbm{Z}}$ obeys the recurrence relation with constant
  coefficients $b (\frac{1}{2} (S + S^{-1})) \cdot
  (y_n) = (a_n)$.
\end{proposition}

\begin{proof}
This is actually the limit case~$r = 0$ of
Theorem~\ref{thm:rec}, but a direct proof is very easy: just write
\[\sum_{i = -\deg b}^{\deg b} b_i z^i  \sum_{n = - \infty}^{\infty} y_n z^n
   = \sum_{n = - \infty}^{\infty} \Bigl(\sum_{i = - \infty}^{\infty} b_i y_{n
   - i}\Bigr) z^n = \sum_{n = - \infty}^{\infty} a_n z^n, \hspace{1em} x =
   \frac{z + z^{-1}}{2}, \]
and identify the coefficients of like powers
of~$z$.
\end{proof}

As in the general case (Section~\ref{sec:series}), this recurrence has spurious (divergent) solutions besides the ones we are interested in.
However, we
can explicitly separate the positive powers of~$z$ from the negative ones in
the Laurent series expansion
\begin{equation}
  \hat{y} (z) = y \Bigl(\frac{z + z^{-1}}{2}\Bigr) = \sum_{n
  = - \infty}^{\infty} y_n z^n, \hspace{2em} \rho^{-1} < \mathopen| z \mathclose| <
  \rho, \label{eq:laurent ratpoly}
\end{equation}
using partial fraction decomposition. From the computational point of view, it
is better to start with the squarefree factorization of the denominator
of~$\hat{y}$:
\begin{equation}
  \beta (z) = z^{\deg b} b \left(\frac{z + z^{-1}}{2}\right) =
  \beta_1 (z) \beta_2 (z)^2 \cdots \beta_k (z)^k \label{eq:sqf}
\end{equation}
and write the full partial fraction decomposition of $\hat{y} (z)$ in the form
\begin{equation}
  \hat{y} (z) = q (z) + \sum_{i = 1}^k
  \sum_{\beta_i (\zeta) = 0} \sum_{j = 1}^i \frac{h_{i, j} (\zeta)}{(\zeta - z)^j},
  \hspace{2em}
  q (z) = \sum_n q_n z^n \in \mathbbm{Q} [z],
  \label{eq:fullparfrac}
\end{equation}
where $h_{i, j} \in \mathbbm{Q} (z)$.
The~$h_{i, j}$ may be computed efficiently using the Bronstein-Salvy
algorithm~{\cite{BronsteinSalvy1993}} (see also~{\cite{GourdonSalvy1996}}).

We obtain an identity of the form~(\ref{eq:laurent ratpoly}) by expanding the
partial fractions corresponding to poles~$\zeta$ with~$\mathopen| \zeta \mathclose| >
1$ in power series about the origin, and those with~$\mathopen| \zeta \mathclose| < 1$
about infinity. The expansion at infinity of
\[\frac{h_{i, j} (\zeta)}{(\zeta - z)^j} =
   \frac{(- 1)^j z^{-j} h_{i, j} (\zeta)}{(1
   - \zeta z^{-1})^j} \]
converges for $\abs{z} > \abs{\zeta}$ and
does not contribute to the coefficients of~$z^n$, $n \geqslant 0$ in the
complete Laurent series. It follows from the uniqueness of the Laurent
expansion of~$\hat{y}$ on the annulus $\rho^{-1} < \mathopen| z \mathclose| < \rho$
that{\footnote{To prevent confusion, it may be worth pointing out that in the
expression
\[\hat{y} (z) = q (z) + \sum_{i = 1}^k
   \sum_{\text{\scriptsize{$\begin{array}{c}
     \beta_i (\zeta) = 0\\
     \mathopen| \zeta \mathclose| > 1
   \end{array}$}}} \sum_{j = 1}^i (\frac{h_{i, j} (\zeta
)}{(\zeta - z)^j} + \frac{h_{i, j} (\zeta^{-1}
)}{(\zeta^{-1} - z)^j}) \]
the Laurent expansion of a single term of the form $\frac{h_{i, j} (
\zeta)}{(\zeta - z)^j} + \frac{h_{i, j} (\zeta^{-1}
)}{(\zeta^{-1} - z)^j}$ is {{\em not\/}} symmetric for~$j
> 1$, even if~$q (z) = 0$.}}
\begin{equation}
  \sum_{n = 0}^{\infty} y_n z^n = q (z) + \sum_{i = 1}^k
  \sum_{\text{\scriptsize{$\begin{array}{c}
    \beta_i (\zeta) = 0\\
    \mathopen| \zeta \mathclose| > 1
  \end{array}$}}} \sum_{j = 1}^i \frac{h_{i, j} (\zeta)}{(
  \zeta - z)^j} . \label{eq:sum zeta<gtr>0}
\end{equation}
We now extract the coefficient of~$z^n$ in~(\ref{eq:sum zeta<gtr>0}) and use
the symmetry of~$(y_n)_{n \in \mathbbm{Z}}$ to get an explicit
expression of~$y_n$ in terms of the roots of~$b (\frac{1}{2}  (z +
z^{-1}))$.

\begin{proposition}
  \label{prop:chebcoeff ratpoly}The coefficients of the Chebyshev expansion~$y
  (x) = \sum_n y_{\mathopen| n \mathclose|} T_n
  (x)$ are given by
  \begin{equation}
    y_n = q_n + \sum_{i = 1}^k \sum_{j = 1}^i
    \sum_{\text{\scriptsize{$\begin{array}{c}
      \beta_i (\zeta) = 0\\
      \mathopen| \zeta \mathclose| > 1
    \end{array}$}}} \binom{n + j - 1}{j - 1} h_{i, j} (\zeta)
    \zeta^{-n - j}  \hspace{2em} (n \geqslant 0)
    \label{eq:chebcoeff ratpoly}
  \end{equation}
  where the~$q_n \in \mathbbm{Q}$, $\beta_i \in \mathbbm{Q} [z]$
  and~$h_{i, j} \in \mathbbm{Q} (z)$ are defined in
  Equations~(\ref{eq:sqf}) and~(\ref{eq:fullparfrac}).
\end{proposition}

Note that~(\ref{eq:sum zeta<gtr>0}) also yields a recurrence of order~$\deg b$
on~$(y_n)_{n \in \mathbbm{N}}$, instead of $2 \deg b$ for that
from Proposition~\ref{prop:chebrec ratpoly}, but now with algebraic instead of rational coefficients in general.

\subsection{Truncation Error}

We can now explicitly bound the error in truncating the Chebyshev expansion of~$y$.

\begin{proposition}
  \label{prop:tailbound ratpoly}Let~$y \in \mathbbm{Q} (x)$ have
  no pole within the elliptic disk~$E_{\rho}$ (see~(\ref{eq:ellipses})).
  Assume again the notations from~(\ref{eq:sqf}) and~(\ref{eq:fullparfrac}).
  For all~$d \geqslant \deg q$, it holds that
  \[ \Bigl\| \sum_{n > d} y_n T_n \Bigr\|_{\infty} \leqslant \sum_{i = 1}^k
     \sum_{j = 1}^i \sum_{\text{\scriptsize{$\begin{array}{c}
       \beta_i (\zeta) = 0\\
       \mathopen| \zeta \mathclose| > 1
     \end{array}$}}} \frac{\mathopen| h_{i, j} (\zeta) \mathclose| 
     (d + 2)^{j - 1} }{(\mathopen| \zeta \mathclose| - 1)^j}
     \mathopen| \zeta \mathclose|^{-d - 1}   = O (d^{\deg b} \rho^{-d})  \]
  as $d \to \infty$.
\end{proposition}

\begin{proof}
We have $\| \sum_{n > d} y_n T_n
\|_{\infty} \leqslant \sum_{n > d} \mathopen| y_n \mathclose|$ because $\|
T_n \|_{\infty} \leqslant 1$ for all~$n$. Using the inequality
\[\sum_{n > d} \binom{n + j - 1}{j - 1} t^{n + j} \leqslant (d + 2
)^{j - 1} t^{d + 1}  \sum_{n = 0}^{\infty} \binom{n + j - 1}{j - 1}
   t^{n + j} = \frac{(d + 2)^{j - 1} t^{d + j + 1}}{(1 - t
)^j} \]
for $t<1$,
the explicit expression from Proposition~\ref{prop:chebcoeff ratpoly} yields
\begin{align*}
     \sum_{n > d} \mathopen| y_n \mathclose| & \nosymbol \leqslant \sum_{n > d}
     \sum_{i = 1}^k \sum_{j = 1}^i \sum_{\substack{
       \beta_i (\zeta) = 0\\
       \mathopen| \zeta \mathclose| > 1
     }} \binom{n + j - 1}{j - 1}  \mathopen| h_{i, j} (\zeta
  ) \mathclose|  \mathopen| \zeta \mathclose|^{-n - j} \\
     & \nosymbol \leqslant \sum_{i = 1}^k \sum_{j = 1}^i
     \sum_{\substack{
       \beta_i (\zeta) = 0\\
       \mathopen| \zeta \mathclose| > 1
     }} \frac{\mathopen| h_{i, j} (\zeta) \mathclose|  (
     d + 2)^{j - 1} }{(\mathopen| \zeta \mathclose| - 1)^j} 
     \mathopen| \zeta \mathclose|^{-d - 1}  .
   \end{align*} 
Since~$\mathopen| \zeta \mathclose| > 1$ actually implies~$\mathopen| \zeta \mathclose| > \rho$ when $b (\frac{1}{2} {(\zeta + \zeta^{-1})}) = 0$, the asymptotic estimate follows.
\end{proof}

\subsection{Computation}

There remains to check that the previous results really translate into a
linear time algorithm.
We first state two lemmas regarding basic operations with polynomials written on the Chebyshev basis.

\begin{lemma}
  \label{lemma:mult}
  The product $a b$ where the operands $a, b \in \Q[x]$ and the result are written in the Chebyshev basis may be computed in $O((\deg a) (\deg b))$ operations.
\end{lemma}

\begin{proof}
It suffices to loop over the indices $(i,j)$ and, at each step, add to the coefficients of $T_{\abs{i \pm j}}$ in the product the contribution coming from the coefficient of~$T_i$ in~$a$ and that of~$T_j$ in~$b$, according to the formula $2 T_i T_j = T_{i+j} + T_{i-j}$.
\end{proof}

The naïve Euclidean division algorithm~{\cite[Algorithm~2.5]{GathenGerhard2003}} runs in linear time with respect to the degree of the dividend when the divisor is fixed.
Its input and output are usually represented by their coefficients in the monomial basis, but the algorithm is easily adapted to work in other polynomial bases.

\begin{lemma}
  \label{lemma:divrem}
  The division with remainder $a = bq + r$ ($\deg r < \deg b$) where $a, b, q, r \in \mathbbm{Q} [x]$ are represented in the Chebyshev basis may be performed in~$O (\deg a)$ operations for fixed~$b$.
\end{lemma}

\begin{proof}
    Assume $n = \deg a > \deg b = m$. The classical
polynomial division algorithm mainly relies on the fact that $\deg (a -
b_m^{-1} a_n x^{n - m} b) < n$ where $a = \sum_i a_i x^i$ and $b =
\sum_i b_i x^i$. From the multiplication formula $2 T_n T_m = T_{n + m} + T_{n
- m}$ follows the analogous inequality $\deg (a - 2 b_{m^{}}^{-1} a_n
T_{n - m} b) < n$ where~$a_k$, $b_k$ now denote the coefficients of~$a$
and~$b$ in the Chebyshev basis. Performing the whole computation in that basis
amounts to replacing each of the assignments $a \leftarrow a - b_m^{-1} a_n
x^{n - m} b$ repeatedly done by the classical algorithm by $a \leftarrow a - 2
b_{m^{}}^{-1} a_n T_{n - m} b$. Since the polynomial $T_{n - m} b$ has at
most~$2 m$ nonzero coefficients, each of these steps takes constant time with
respect to~$n$. We do at most~$n - m$ such assignments, hence the overall
complexity is~$O (n)$.
\end{proof}

We end up with Algorithm~\ref{algo:RationalExpansion}.
In view of future needs, it takes as input a polynomial of arbitrary degree already written in the Chebyshev basis in addition to the rational function (of bounded degree)~$y$. The details of the algorithm are only intended to support the complexity estimates, and many improvements are possible in practice.

\begin{algorithm}
    \label{algo:RationalExpansion}

    \emph{Input:}
        a rational fraction $y(x)=a(x)/b(x)$,
        the Chebyshev coefficients of a polynomial $f=\sum_{n=-d}^d f_n T_n (x)$,
        an error bound~$\epsilon$.
    \emph{Output:}
        the Chebyshev coefficients of an approximation $\tilde{y}(x)$ of~$fy$ such that $\ninf{\tilde{y}-fy}\leqslant\epsilon$.

    \everymath{\displaystyle}
    \setcounter{algoline}{0}
    \begin{body}
    \labelitem{step:convert} convert $a$ and $b$ to Chebyshev basis
    \labelitem{step:multiplication}compute the polynomial $g=af$, working in the Chebyshev basis
    \labelitem{step:euclidean} compute the quotient $q$ and the remainder $r$ in the Euclidean division of $g$ by~$b$
    \labelitem{step:parfrac} compute the partial fraction decomposition of $\hat{w}(z)=w(x)=r(x)/b(x)$, where $x=\frac12(z+z^{-1})$, using the Bronstein-Salvy algorithm
    \labelitem{step:degree} find $d'\geqslant \deg q$ such that $\Bigl\|\sum_{n>d'}y_nT_n\Bigr\| \leqslant \epsilon/4$ using Proposition~\ref{prop:tailbound ratpoly}
    \item compute $\rho_{-}$ and $\rho_{+}$ such that $\beta_i (\zeta) = 0 \wedge \mathopen| \zeta \mathclose| > 1 \Rightarrow 1<\rho_{-}\leqslant\mathopen| \zeta \mathclose|\leqslant \rho_{+}$
    \item compute $M \geqslant \sum_{i = 1}^k \sum_{j = 1}^i j (\deg \beta_i) \sup\limits_{\rho_{-}\leqslant\mathopen| \zeta \mathclose|\leqslant \rho_{+}}\left(|h'_{i,j}(\zeta)| + |\zeta^{-1} h_{i,j}(\zeta)|\right) \rho_{-}^{-j}$
    \labelitem{step:epsilonPrime} set $\epsilon':= \min{\left(\rho_{-}-1, M^{-1}\left(1-\rho_{-}^{-1}\right)^{D+1}\frac{\epsilon}{4}\right)}$, with $D=\deg b$
    \labelitem{step:roots} compute approximations $\tilde{\zeta}\in\Q[i]$ of the roots $\zeta$ of $\beta_i$ such that $\abs{\tilde{\zeta}-\zeta} < \epsilon'$
    \item for $0\leqslant n \leqslant d'$
        \begin{body}
        \item set $\tilde{y}_n=q_n+ \operatorname{Re}\Biggl(\sum_{i = 1}^k\sum_{j = 1}^i
            \sum_{
            \substack{ \beta_i (\zeta) = 0\\ \mathopen| \zeta \mathclose| > 1 }}
            \binom{n + j - 1}{j - 1} h_{i, j} (\tilde{\zeta}) \tilde{\zeta}^{-n - j}\Biggr)$
        \end{body}
    \item return $\tilde{y}(x)=\sum\limits_{n=-d'}^{d'}\tilde{y}_n T_n(x)$
    \end{body}
\end{algorithm}

\begin{proposition}
    \label{prop:RationalExpansion}
    Algorithm~\ref{algo:RationalExpansion} is correct.
    As $d \to \infty$ and $\epsilon \to 0$ with all other parameters fixed, it runs in $O(d+\ln(\epsilon^{-1}))$ arithmetic operations and returns a polynomial of degree~$d' \leq \max(d, K \ln(\epsilon^{-1}))$, where $K$~depends on~$y$, but not on $f$~or~$d$.
\end{proposition}

\begin{proof}
  Firstly, we prove the error bound $\ninf{\tilde{y}-y}\leqslant \epsilon$.
  Let $A=\{\zeta: \rho_{-}\leqslant\mathopen| \zeta \mathclose|\leqslant \rho_{+}\}$ and
  \[ M_0=\sup\limits_{\zeta \in A} |h'_{i,j}(\zeta)|, \quad M_1=\sup\limits_{\zeta \in A} |\zeta^{-1} h_{i,j}(\zeta)|. \]
  On~$A$, we have
  \begin{equation*}
    \abs{\frac{\mathd}{\mathd \zeta} \left(h_{i, j}(\zeta) \, \zeta^{-n - j} \right)}
    \leqslant  \left( M_0+(n+j)M_1\right) | \zeta|^{-n-j}\leqslant (n+j)(M_0+M_1)\rho_{-}^{-n-j}.
  \end{equation*}
  By Proposition~\ref{prop:chebcoeff ratpoly}, observing that the condition $|\zeta-\tilde{\zeta}|<\rho_{-}-1$ from Step~\ref{step:epsilonPrime} implies $[\zeta, \tilde{\zeta}] \subset A$, we have the inequalities
  \begin{align*}
    |y_n-\tilde{y}_n| & \nosymbol \leqslant \sum_{i = 1}^k \sum_{j = 1}^i
    \sum_{\substack{ \beta_i (\zeta) = 0\\ \mathopen| \zeta \mathclose| > 1}}
    \binom{n + j - 1}{j - 1} \left| h_{i, j} (\zeta) \zeta^{-n - j} -h_{i, j} (\tilde{\zeta}) \tilde{\zeta}^{-n - j}  \right| \\
    & \leqslant \sum_{i = 1}^k \sum_{j = 1}^i j \, (\deg \beta_i) \binom{n + j}{j} (M_0 + M_1) \rho_{-}^{-n-j}\epsilon'
    \leqslant M \binom{n + D}{D} \rho_{-}^{-n}\epsilon'.
  \end{align*}
  Therefore, the output of Algorithm~\ref{algo:RationalExpansion} satisfies
  \begin{align*}
    \ninf{y_n-\tilde{y}_n}  \leqslant \sum_{n = -d'}^{d'}|y_n-\tilde{y}_n| + 2\, \Bigl\|{\sum_{n>d'}y_nT_n}\Bigr\| \leqslant \dfrac{2M \epsilon'}{(1-\rho_{-}^{-1})^{D+1}}+2\dfrac{\epsilon}{4}
    \leqslant \epsilon.
  \end{align*} 
  By Proposition~\ref{prop:tailbound ratpoly}, for all $\epsilon > 0$, there exists $d' \leq K \ln(\epsilon^{-1})$ with
  ${\bigl\| \sum_{n>d} y_n T_n \bigr\| \leq \epsilon}$.
  In addition, we have $\deg q \leq d$, hence the degree of the output, as computed in Step~\ref{step:degree}, satisfies $d' \leq \max(d, K \ln(\epsilon^{-1}))$.

  Turning to the complexity analysis, Steps~\ref{step:convert} and \ref{step:parfrac}--\ref{step:epsilonPrime} have constant cost.
  So does each iteration of the final loop, assuming the powers of $\tilde\zeta^{-1}$ are computed incrementally.
  Steps \ref{step:multiplication}~and~\ref{step:euclidean} take~$O(d)$ operations by Lemmas~\ref{lemma:mult}~and~\ref{lemma:divrem}, and do not depend on~$\epsilon$.
  Regarding Step~\ref{step:roots}, it is known~\cite[Theorem 1.1(d)]{Pan1996} that the roots of a polynomial with integer coefficients can be approximated with of absolute accuracy $\eta$ in $O(\eta^{-1})$ arithmetic operations.
  (In fact, the bit complexity of the algorithm beyond this statement is also softly linear in $\ln (\epsilon^{-1})$.)
  Since $M$~depends neither on~$\epsilon$ nor on~$d$, we have $\epsilon'=\Omega(\epsilon)$, and hence the cost of Step~\ref{step:roots} is in $O(\ln(\epsilon^{-1}))$.
\end{proof}

\section{Validation}
\label{sec:validation}

Assume that we have computed a polynomial of degree~$d$
\[ p ( x) = \sum_{n = - d}^d \tilde{y}_n T_n ( x) \]
which presumably is a good approximation on $[ - 1, 1]$ of the D-finite
function~$y$ defined by~(\ref{eq:deq},\ref{eq:inicond}). As stated in
Problem~\ref{problem}, our goal is now to obtain a reasonably tight bound~$B$
such that $\| y - p \|_{\infty} \leqslant B$.

\subsection{Principle}The main idea to compute the bound is to convert the
initial value problem defining~$y$ into a fixed-point equation~$T ( y) = y$,
verify explicitly that $T$~maps some neighborhood of~$p$ into itself, and
conclude that $p$ must be close to the ``true'' solution. This is one simple
instance of the functional enclosure methods mentioned in
Section~\ref{sec:background} and widely used in interval
analysis~{\cite{Moore1979}}. In most cases (with the notable exception
of~{\cite{EpsteinMirankerRivlin1982a,KaucherMiranker1984,KaucherMiranker1988}}),
these methods are based on interval Taylor series expansions such as so-called
Taylor models~{\cite{MakinoBerz2003,Neumaier2003}}. Here, we describe an
adaptation designed to work {\emph{in linear time}} with polynomial
approximations written on the Chebyshev basis.

The first step is to reduce the differential initial value problem
defining~$y$ to a linear integral equation of the Volterra type and second
kind. Let $\alpha_0, \alpha_1, \ldots, \alpha_r \in \mathbbm{Q} [ x]$ be such
that
\[ L = a_r \partial^r + \cdots + a_1 \partial + a_0 = \partial^r \alpha_r +
   \cdots + \partial \alpha_1 + \alpha_0, \]
and, for $k \in \llbracket 0, r \rrbracket$, define
\[ L_k = \partial^{r - k} \alpha_r + \cdots + \partial \alpha_{k + 1} +
   \alpha_k \in \mathbbm{Q} [ x] \langle \partial \rangle . \]
Observe that for all $a \in \mathbbm{Q} [ x]$ and $k \in \mathbbm{N}$, the operator $\partial^k a$ is of the form
\[ \partial^k a = a \partial^k + \ast \partial^{k - 1} + \cdots, \]
so that $\alpha_r=a_r$.
In particular, $\alpha_r$~does not vanish on $[ - 1, 1]$.

\begin{lemma}
  \label{lemma:Volterra}The solution $y \in X$ of $L \cdot y = 0$ such that
  $y^{( k)} ( 0) = \ell_k$ for $k \in \llbracket 0, r - 1 \rrbracket$
  satisfies
  \begin{equation}
    \alpha_r ( x)  \text{} y ( x) = g ( x) + \int_0^x K ( x, t) y ( t) \mathd
    t, \label{eq:Volterra}
  \end{equation}
  where
  \begin{equation}
    K ( x, t) = - \sum^{r - 1}_{k = 0} \frac{( x - t)^k}{k!} \alpha_{r - 1 -
    k} ( t) = \sum_{j = 0}^{r - 1} \beta_j ( t) x^j, \label{eq:K}
  \end{equation}
  \[ \beta_j ( t) = \sum_{i = 0}^{r - 1 - j} \frac{( - 1)^{i + 1}}{i!j!} t^i
     \alpha_{r - 1 - j - i} ( t), \]
  and
  \begin{equation}
    g ( x) = \sum_{k = 0}^{r - 1} ( L_{r - k} \cdot y) ( 0)  \frac{x^k}{k!} =
    \sum_{k = 0}^{r - 1} \sum_{j = 0}^k \sum_{i = j}^k \binom{i}{j} \alpha_{r
    - k + i}^{( i - j)} ( 0) \ell_j  \frac{x^k}{k!} . \label{eq:g}
  \end{equation}
\end{lemma}

\begin{proof}
  We have
  \[ \int_0^x L_k \cdot y = ( L_{k + 1} \cdot y) ( x) - ( L_{k + 1} \cdot y) (
     0) + \int_0^x \alpha_k y. \]
  Applying this formula to integrate $r$~times the equation $L \cdot y = 0$
  yields
  \begin{align*}
    &
    ({\alpha}_{r}y)(x)+\int_{0}^{x}({\alpha}_{r-1}y)(x_{1}){\mathd}x_{1}+{\cdots}+\int_{0}^{x}\int_{0}^{x_{1}}{\cdots}\int_{0}^{x_{r-1}}({\alpha}_{0}y)(x_{r}){\mathd}x_{r}{\cdots}{\mathd}x_{2}{\mathd}x_{1}\\
     &=
    (L_{1}{\cdot}y)(0){\frac{x^{r-1}}{(r-1)!}}+{\cdots}+(L_{r-1}{\cdot}y)(0)x+(L_{r}{\cdot}y)(0),
  \end{align*}
  and~(\ref{eq:Volterra}) follows using the relation
  \[ \int_0^x \int_0^{x_1} \cdots \int_0^{x_{k - 1}} f ( x_k) \mathd x_k
     \cdots \mathd x_2 \mathd x_1 = \int_0^x \frac{( x - t)^{k - 1}}{( k - 1)
     !} f ( t) \mathd t \]
  (which can be obtained by repeated integrations by parts).
\end{proof}

Let $X$ denote the Banach space of continuous functions from $[ 0, 1]$ to
$\mathbbm{R}$ (or $\mathbbm{C}$), equipped with the uniform norm. With
$K$~and~$g$ as in Lemma~\ref{lemma:Volterra}, define an operator $T : X
\rightarrow X$ by
\begin{equation}
  T ( f) ( x) = \alpha_r ( x)^{- 1}  \left( g ( x) + \int_0^x K ( x, t) f ( t)
  \mathd t \right), \label{eq:T}
\end{equation}
so that~(\ref{eq:Volterra}) becomes $T ( y) = y$. The next proposition is an
instance of a classical bound (compare, e.g.,
Rall~{\cite[Chap.~1]{Rall1969}}).
As we recall below, when $p$~is close to the fixed point~$y$, its iterated image~$T^i(p)$ remains close to~$y$, which yields good upper bounds on the distance $\|y-p\|_{\infty}$.

\begin{proposition}
  \label{prop:invbound}Assume that $A$ is an upper bound on $| \alpha_r (
  x)^{- 1} K ( x, t) |$ for $x \in [ - 1, 1]$ and $t$ between $0$~and~$x$.
  Then the bound
  \[ \| p - y \|_{\infty} \leqslant \gamma_i  \| T^i ( p) - p \|_{\infty},
     \hspace{2em} \gamma_i = \sum_{j = 0}^{\infty} \frac{A^{ij}}{( ij) !}, \]
  holds for all $i \geqslant 1$.
\end{proposition}

\begin{proof}
  Let~$V : f \mapsto \alpha_r^{- 1} \int Kf$ denote the linear part of~$T$.
  Equation~\eqref{eq:T} rewrites again as $(\operatorname{Id} - V) \cdot y = g$. More
  generally, we have
  \[ (\operatorname{Id} - V^i) \cdot y = T^i (y) - V^i \cdot y = T^i (f) - V^i
     \cdot f = (\operatorname{Id} + V + \cdots + V^{i - 1}) \cdot g \]
  for all $f \in X$. The operator~$V$ is continuous, and since
  \[ | V^i \cdot y (x) | \leqslant \int_0^x A \int_0^{x_1} A \cdots
     \int_0^{x_{i - 1}} A \| y \|_{\infty} \mathd x_i \cdots \mathd x_2 \mathd
     x_1 = \frac{A^i}{i!}  \| y \|_{\infty}, \]
  the (subordinate) norm of $V^i$ satisfies
  \[ \| V^i \| \leqslant \frac{A^i}{i!} . \]
  Therefore, the series $\sum_j V^{ij}$ converges, $(\operatorname{Id} - V^i)$ is
  invertible, and $\| (\operatorname{Id} - V^i)^{- 1} \| \leqslant \gamma_i$. Writing
  \[ p - y = (\operatorname{Id} - V^i)^{- 1} \cdot [ (\operatorname{Id} - V^i) \cdot p - (
     T^i (p) - V^i \cdot p)] = (\operatorname{Id} - V^i)^{- 1} \cdot (p - T^i (p))
  \]
  yields the announced result.
\end{proof}

The interesting fact about this bound is that we can effectively compute
rigorous approximations of the right-hand side. Indeed, $p$~and~$g$ are
explicit polynomials, so that it is not too hard to compute $T^i (p) - p$
approximately while keeping track of the errors we commit, and deduce a bound
on its norm.

Choosing $i = 1$ (and $\gamma_1 = e^A$) in Proposition~\ref{prop:invbound}
already yields a nontrivial estimate, but (as we shall see in more detail)
larger values of~$i$ are useful since $\gamma_i \rightarrow 1$ as $i
\rightarrow \infty$. In particular, we have
\begin{equation}
  \gamma_i \leqslant \frac{1}{1 - A^i / i!} \label{eq:bound on gamma}
\end{equation}
for large~$i$. This crude bound will be enough for our purposes. In practice,
though, it is better to compute an approximation of $\gamma_i$ that is closer
to reality (e.g., using the first few terms of the series and a bound on the
tail) in order to reduce the required number of iterations of~$V$.

\subsection{Algorithm}
\label{sec:validation-algorithm}

The actual computation of $\| T^i (p) - p \|_{\infty}$ relies on the following lemmas.

\begin{lemma}
  \label{lemma:int}One can compute an antiderivative of a
  polynomial{\footnote{Note however that derivation does not commute with the
  truncation of Chebyshev {\emph{series}}.}} of degree at most~$d$ written on
  the Chebyshev basis in~$O (d)$ arithmetic operations.
\end{lemma}

\begin{proof}
  If $f = \sum_n c_n T_n$ and $f' = \sum_n c'_n T_n$, then we have $2 nc_n =
  c'_{n - 1} + c'_{n + 1}$ according to Equation~\eqref{eq:mixedrelTn}.
\end{proof}

\begin{lemma}
  \label{lemma:boundpoly}Let $f (x) = \sum_{n = - d}^d c_n T_n (x)$ (with
  $c_n = c_{- n}$) be a polynomial of degree~$d$, given on the Chebyshev
  basis. Then one can compute $M \geqslant 0$ such that
  \[ \| f \|_{\infty} \leqslant M \leqslant \sqrt[]{d + 1}  \| f \|_{\infty}
  \]
  in $O (d)$ arithmetic operations.
\end{lemma}

\begin{proof}
  It suffices to take $M = \sum_{n = - d}^d | c_n |$. Indeed, we have
  \[ \| f \|_2 = \left(\frac{1}{\pi}  \int_{- 1}^1 \frac{f (t)^2}{\sqrt[]{1
     - t^2}} \mathd t \right)^{1 / 2} = \left(| c_0 |^2 + 4 \sum_{n = 1}^d |
     c_n |^2 \right) . \]
  It follows that
  \[ \| f \|_2 \leqslant \| f \|_{\infty} \leqslant M \leqslant \sqrt[]{d + 1}
     \| f \|_2, \]
  where the first inequality results from the integral expression of~$\| f
  \|_2$, the second one, from the fact that $\| T_n \|_{\infty} \leqslant 1$
  for all~$n$, and the last one from the Cauchy-Schwarz inequality.
\end{proof}

This results in Algorithm~\ref{algo:Validation}. Again here, the suggested
bounds for $\| \alpha_r^{- 1} \|_{\infty}$, $\| p - p_i \|_{\infty}$ and
$\gamma_i$ are only intended to support the complexity estimate, and tighter
choices are possible in practice.

\begin{algorithm} \label{algo:Validation}
\setcounter{algoline}{0}
\emph{Input:}
   A differential operator $L = a_r \partial^r + \cdots + a_1
  \partial + a_0 \in \mathbbm{Q} [ x] \langle \partial \rangle$ of order~$r$
  such that $a_r (x) \neq 0$ for $x \in [ - 1, 1]$, initial values $\ell_0,
  \ell_1, \ldots, \ell_{r - 1}$. A degree\mbox{-}$d$ polynomial~$p$ written on
  the Chebyshev basis.
  An accuracy parameter ~$\varepsilon > 0$.
    \emph{Output:}
  A real number $B > 0$ such that $\| y - p \|_{\infty}
  \leqslant B$, where $y$~is the unique solution of $L \cdot y = 0$ satisfying
  $y (0) = \ell_0, y' (0) = \ell_1, \ldots, y^{(r - 1)} (0) = \ell_{r -
  1}$.

\begin{body}
  \labelitem{step:Valid:commutations}using the commutation rule $x \partial
  = \partial x - 1$, compute polynomials $\alpha_0, \alpha_1, \ldots, \alpha_r
  \in \mathbbm{Q} [ x]$ such that $L = \partial^r \alpha_r + \cdots + \partial
  \alpha_1 + \alpha_0$
  
  \item define $K \in \mathbbm{Q} [ t, x]$, $(\beta_j)_{j = 0}^{r - 1} \in
  \mathbbm{Q} [ x]^r$ and $g \in \mathbbm{Q} [ x]$ as in
  Lemma~\ref{lemma:Volterra}
  
  \labelitem{step:Valid:A}compute $A \geqslant \max \left\{ | \alpha_r (
  x)^{- 1} K (x, t) | : \text{$0 \leqslant t \leqslant x \leqslant 1$ or $- 1
  \leqslant x \leqslant t \leqslant 0$} \right\}$ (e.g., using
  Algorithm~\ref{algo:RationalExpansion} to expand $\alpha_r (x)^{- 1}$ in
  Chebyshev series), and define $(\gamma_i)_{i = 1}^{\infty}$ as in
  Proposition~\ref{prop:invbound}
  
  \labelitem{step:Valid:choose i}   compute the minimum $i$ such that~$A^i / i! \leqslant 1 / 2$
  
  \item set $p_0 = p$
  
  \item for $k = 0, 1, \ldots, i - 1$
  \begin{body}
    \labelitem{step:Valid:qk}compute $q_{k + 1} (x) = g (x) + \sum_{j =
    0}^{r - 1} x^j  \int_0^x \beta_j (t) p_k (t) \mathd t \in \mathbbm{Q} [
    x]$
    
    \labelitem{step:Valid:RationalExpansion}
    compute $p_{k + 1} \in \Q[x]$ such that
    ${\| p_{k + 1} - \alpha_r (x)^{- 1} q_{k + 1} (x) \|_{\infty}} \leqslant \varepsilon$
    using Algorithm~\ref{algo:RationalExpansion}
  \end{body}
  \labelitem{step:Valid:boundpoly}compute $\delta \geqslant \| p - p_i
  \|_{\infty}$ using Lemma~\ref{lemma:boundpoly}
  
  \item return $B = \gamma_i  (\delta + e^A \varepsilon)^{}$
\end{body}
\end{algorithm}

We now prove that Algorithm~\ref{algo:Validation} works as stated,
and estimate how tight the bound it returns is.

\begin{theorem}
  \label{thm:Valid}Algorithm~\ref{algo:Validation} is correct: its output~$B$
  is an upper bound for $\| y - p \|_{\infty}$.
  For fixed~$L$,
  as $d \rightarrow \infty$ and $\varepsilon \rightarrow 0$,
  the bound~$B$ satisfies
  \[ B = O\bigl( (\|y-p\|_{\infty} + \epsilon) (d+\ln(\epsilon^{-1}))^{1/2} \bigr) \]
  and the algorithm performs
  $O (d + \ln (\varepsilon^{- 1}))$
  arithmetic operations.
\end{theorem}

\begin{proof}
  Denote by~$V$ the linear part of~$T$ and recall from the proof of
  Proposition~\ref{prop:invbound} that $\| V^i \| \leqslant A^i / i!$. For
  all~$k$, the polynomial~$p_{k + 1}$ computed on
  line~\ref{step:Valid:RationalExpansion} satisfies
  \[ \| p_{k + 1} - T (p_k) \|_{\infty} \leqslant \varepsilon, \] hence we have
  \begin{align*}
    &\|p_{i}-T^{i}(p)\|_{{\infty}} \\
    &
    \leq\|p_{i}-T(p_{i-1})\|_{{\infty}}+\|T(p_{i-1})-T^{2}(p_{i-2})\|_{{\infty}}+{\cdots}+\|T^{i-1}(p_{1})-T^{i}(p_{0})\|_{{\infty}}\\
    &
    \leq\|p_{i}-T(p_{i-1})\|_{{\infty}}+\|V\|\|p_{i-1}-T(p_{i-2})\|_{{\infty}}+{\cdots}+\|V^{i-1}\|\|p_{1}-T(p_{0})\|_{{\infty}}\\
    & \leq e^{A}{\varepsilon}.
  \end{align*}
  By Proposition~\ref{prop:invbound}, it follows that
  \begin{equation}
    \label{eq:p-y delta}
    \| p - y \|_{\infty} \leqslant \gamma_i  \| p - T^i (p) \|_{\infty}
    \leqslant \gamma_i  (\| p - p_i \|_{\infty} + \| p_i - T^i (p)
    \|_{\infty}) \leqslant \gamma_i  (\delta + e^A \varepsilon) .
  \end{equation}
  This establishes the correctness of the algorithm.

  We now turn to the tightness statement.
  Letting $D = \deg(p-p_i) + 1$,
  Lemma~\ref{lemma:boundpoly} implies that
  \begin{equation}
    \label{eq:delta}
    \delta
    \leq \sqrt D \, \| p - p_i \|_{\infty}
  \end{equation}
  where
  \begin{equation}
    \label{eq:delta-ppi}
    \begin{aligned}
      \| p - p_i \|_{\infty}
      & \leq \| p - y \|_{\infty}
          + \| y - T^i(p) \|_{\infty}
          + \| T^i(p) - p_i \|_{\infty} \\
      & \leq \bigl(1 + \| V^i \| \bigr) \| p - y \|_{\infty}
          + e^A \epsilon \\
      & \leq \frac32 \| p - y \|_{\infty} + e^A \epsilon.
    \end{aligned}
  \end{equation}
  Looking at the definition of~$p_{k+1}$ in step~\ref{step:Valid:qk}, we see that $\deg q_{k+1} \leq \deg p_k + C_1$ for some $C_1>0$ depending on~$L$ only.
  Additionally, according to Proposition~\ref{prop:RationalExpansion}, there exists $C_2$ (again depending on~$L$ only) such that
  $
    \deg p_{k+1} \leq \max\bigl(\deg q_k, C_2 \ln (\epsilon^{-1}) \bigr).
  $
  It follows by induction that
  $
    \deg p_k \leq \max\bigl(d, C_2 \ln (\epsilon^{-1}) \bigr) + C_1 k
  $
  for all~$k$,
  whence
  \begin{equation}
    \label{eq:delta-deg}
    D \leq \max\bigl(d, C_2 \ln (\epsilon^{-1}) \bigr) + C_1 i + 1.
  \end{equation}
  Plugging \eqref{eq:delta-ppi}~and~ \eqref{eq:delta-deg} into~ \eqref{eq:delta} yields the estimate
  \begin{equation} \label{eq:bound delta}
    \delta
    \leq \sqrt D \left( \frac32 \|p-y\|_{\infty} + e^A \epsilon \right)
    = O\bigl( (\|y-p\|+\epsilon)(d+\ln(\epsilon^{-1}))^{1/2} \bigr)
  \end{equation}
  and the result then follows from the definition of~$\delta$ since $\gamma_i \leq 2$.

  Finally, the only steps whose cost depends on
  $d$~or~$\varepsilon$ are lines \ref{step:Valid:qk},
  \ref{step:Valid:RationalExpansion}, and \ref{step:Valid:boundpoly} (and the
  number of loop iterations does not depend on these parameters either). By
  Proposition~\ref{prop:RationalExpansion}, the degrees of $p_k$ and $q_k$ are
  all in $O (d + \ln (\varepsilon^{- 1}))$. The cost of
  step~\ref{step:Valid:qk} is linear in this quantity by Lemmas
  \ref{lemma:mult}~and~\ref{lemma:int}. The same goes for
  line~\ref{step:Valid:RationalExpansion} by
  Proposition~\ref{prop:RationalExpansion}, and for
  line~\ref{step:Valid:boundpoly} by Lemma~\ref{lemma:boundpoly}.
\end{proof}

Another way to put this is to say that Algorithm~\ref{algo:Validation} can be modified to provide an enclosure of~$\|p-y\|_{\infty}$.
Indeed, Equations \eqref{eq:delta}~and~\eqref{eq:delta-ppi} imply
\begin{equation}
  \label{eq:lower bound}
  \|p-y\|_{\infty} \geq b = \frac23 \left( \frac{\delta}{\sqrt D} - e^A \epsilon \right),
\end{equation}
and this~$b$ is a computable lower bound for~$\|p-y\|_{\infty}$.
Furthermore, using~\eqref{eq:p-y delta}, we have
$\delta \geq \gamma_i^{-1} \| p - y \| - e^A \epsilon$,
and hence
\[
  b \geq \frac{1}{3 \sqrt D}
         \bigl( \|p-y\|_{\infty} - (\sqrt D + 2) e^A \epsilon \bigr).
\]
Comparing with the upper bound on~$B$ resulting from~\eqref{eq:bound delta}, we deduce
\[
  \frac{B}{b \ln(b^{-1})}
  \leq
  \frac{
    9 D \|p-y\|_{\infty} + 6 (D + \sqrt D) e^A \epsilon
  }{
    \bigl(
      \|p-y\|_{\infty} - (\sqrt D + 2) e^A \epsilon
    \bigr)
    \ln (\|p-y\|_{\infty}^{-1})
  }.
\]
In particular, if we restrict ourselves to polynomials~$p$ satisfying
$\|p-y\|_{\infty} \leq e^{-\Gamma d}$
for some fixed~$\Gamma$,
and if $\epsilon$~is chosen such that
\begin{equation}
  \label{eq:possible eps}
  \|p-y\|_{\infty}^{2E} \leq \epsilon \leq \|p-y\|_{\infty}^E, \qquad E > 1,
\end{equation}
then
$D \leq \max(\Gamma^{-1}, 2 E C_2) \ln \|p-y\|_{\infty}^{-1} + O(1)$
as $\|p-y\|_{\infty}$ tends to zero, so there exists~$K$ (computable as a function of $y$, $\Gamma$, and~$E$) such that
\begin{equation} \label{eq:effective tightness}\
  B \leq K b \ln(b^{-1}) \leq K \|p-y\|_{\infty} \ln(\|p-y\|_{\infty}^{-1})
\end{equation}
for small~$\|p-y\|_{\infty}$.
Of course, since $\|p-y\|_{\infty}$ is what we want to estimate, we do not know the ``correct'' choice of~$\epsilon$ beforehand.
But, assuming $\|p-y\|_{\infty}$ is indeed small enough, we can search for a suitable~$\epsilon$ iteratively, starting, say, with $\epsilon=2^{-d}$ and checking whether~\eqref{eq:effective tightness} holds at each step.
As our hypotheses imply $d = O(\ln \|p-y\|_{\infty}^{-1})$, the whole process requires at most $O(\ln \|p-y\|_{\infty}^{-1})$ operations.

By combining these tightness guarantees with Corollary~\ref{cor:effective near-minimax} and lower bounds on $\|p^{\ast}_d-y\|_{\infty}$ such as~\eqref{eq:lower bound minimax}, one can devise various strategies to obtain certified polynomial approximations of a given \mbox{D-finite} function~$y$ and relate the computed error bounds to~$\|p^{\ast}_d-y\|_{\infty}$.

\begin{table}[t]
  \begin{center}
    \makebox[0pt]{
    \newcommand{\z}{\phantom0}
    \renewcommand{\tabcolsep}{1ex}
    \begin{tabular}[c]{ccccccrrrrr}
      \toprule
      & & & enclosure of $\ninf{y-p}$ & & & & & \multicolumn{2}{c}{time (s)} \\
      & $d$ & $\epsilon$ & computed by Algo.~\ref{algo:Validation} & $\ninf{y-p}$ & $\ninf{y-p^*}$ & $D$ & $i$ & \ref{algo:Clenshaw} & \ref{algo:Validation} \\
      \midrule
      \multirow{3}{*}{(\ref{ex:ex1AMM})}
      & $30$ & $10^{-104}$  & $[2.3\cdot10^{-53\z},\; 4.3\cdot10^{-52\z}]$ & $3.4\cdot10^{-52\z}$ & $3.4\cdot10^{-52\z}$ & 102  & 2 & $0.05$ & $0.54$\\
      & $60$ & $10^{-194}$  & $[9.0\cdot10^{-99\z},\; 2.4\cdot10^{-97\z}]$ & $2.0\cdot10^{-97\z}$ & $1.9\cdot10^{-97\z}$ & 192 & 2 & $0.05$ & $1.07$\\
      & $90$ & $10^{-284}$  & $[4.6\cdot10^{-144}, \; 1.5\cdot10^{-142} ]$ & $1.2\cdot10^{-142}$  & $1.1\cdot10^{-142}$  & 282 & 2 & $0.06$ & $1.87$\\
      \midrule
      \multirow{3}{*}{(\ref{ex:ex2AMM})}
      & $30$ & $10^{-88\z}$ & $[6.0\cdot10^{-45\z},\; 9.8\cdot10^{-44\z}]$ & $5.9\cdot10^{-44\z}$ & $5.6\cdot10^{-44\z}$ & 42  & 3 & $0.06$ & $0.06$\\
      & $60$ & $10^{-206}$  & $[6.7\cdot10^{-104}, \; 1.5\cdot10^{-102} ]$ & $8.8\cdot10^{-103}$  & $8.5\cdot10^{-103}$  & 72  & 3 & $0.07$ & $0.10$\\
      & $90$ & $10^{-334}$  & $[2.0\cdot10^{-169}, \; 5.1\cdot10^{-168} ]$ & $3.1\cdot10^{-168}$  & $3.0\cdot10^{-168}$  & 102 & 3 & $0.08$ & $0.23$\\
      \midrule
      \multirow{3}{*}{(\ref{ex:ex3AMM})}
      & $30$ & $10^{-18\z}$ & $[1.2\cdot10^{-10},  \; 2.4\cdot10^{-9\z} ]$ & $1.6\cdot10^{-9\z}$  & $1.1\cdot10^{-9\z}$  & 79  & 3 & $0.05$ & $0.74$\\
      & $60$ & $10^{-36\z}$ & $[2.2\cdot10^{-19},  \; 6.1\cdot10^{-18}  ]$ & $4.1\cdot10^{-18}$   & $3.0\cdot10^{-18}$   & 151 & 3 & $0.06$ & $1.6$\\
      & $90$ & $10^{-54\z}$ & $[4.8\cdot10^{-28},  \; 1.7\cdot10^{-26}  ]$ & $1.1\cdot10^{-26}$   & $7.7\cdot10^{-27}$   & 223 & 3 & $0.10$ & $2.7$\\
      \bottomrule
    \end{tabular}
    }
  \end{center}
  \caption{Bounds, parameters appearing in Algorithm~\ref{algo:Validation} and running time of Algorithms \ref{algo:Clenshaw}~and~\ref{algo:Validation} for the examples of Section~\ref{sec:examples-clenshaw}: (\ref{ex:ex1AMM})~$y(x)=e^{x/2}/\sqrt{x+16}$, (\ref{ex:ex2AMM})~$y(x)=\frac32\cos x+\frac12\sin x$, and (\ref{ex:ex3AMM})~$y(x)=(\cos x)/(2x^2+1)$.}
  \label{tab:resultsAMM}
\end{table}

\begin{example}\label{ex:validation}
Table~\ref{tab:resultsAMM} gives validated error bounds obtained for the polynomials computed in Section~\ref{sec:examples-clenshaw} using the code presented there.
In each case, a naïve implementation of Algorithm~\ref{algo:Validation} was called on the polynomial~$p$ returned by Algorithm~\ref{algo:Clenshaw}.
(In the third example, the rough bound~$A$ suggested in Step~\ref{step:Valid:A} was manually replaced by a tighter one to keep the number of iterations small.)
The remaining input parameter $\varepsilon$ was manually set to approximately $\|p-y\|_{\infty}^{2}$
based on a heuristic estimate of $\|p-y\|_{\infty}$.
In practice, this makes the term $e^A \varepsilon$ in the error bound~\eqref{eq:p-y delta} small, so that the main contribution to the error bound~$B$ in practice is $\ninf{p_i-p}$.

Besides the upper bound~$B$, the table gives a lower bound~$b \leq \ninf{p-y}$ obtained as discussed above.
For comparison, we include the ``true'' value of $\ninf{p-y}$, as well as the error $\ninf{y-p^*}$ corresponding to the minimax polynomial of degre~$d$, computed using Sollya~\cite{ChevillardJoldesLauter2010}.

The last four columns indicate the values of the parameters $D$~and~$i$ and the running time of both algorithms.
It can be observed that our choice of~$\varepsilon$ makes $D$~grow significantly larger than~$d$, and that a naïve implementation of Algorithm~\ref{algo:Validation}, despite its interesting theoretical complexity, is far from being efficient in practice.
Nevertherless, for simple examples at least, the total running time remains reasonable.
Note for comparison that plotting the error curves shown on Figure~\ref{fig:plots} is about $1$~to~$2$ times slower than computing the error bounds.
\end{example}

Unfortunately, the above complexity results come short of providing what we may call ``\emph{validated near-minimax} approximations'', at least in a straightforward way.
More precisely,
following Mason and Handscomb~\cite[Def.~3.2]{MasonHandscomb2003}, call an approximation scheme mapping a function~$y$ to a polynomial~$p_d$ of degree at most~$d$ \emph{near-minimax} if it satisfies
\[ \|p_d-y\|_{\infty} \leq \Lambda(d) \, \|p^{\ast}_d-y\|_{\infty} \]
where $\Lambda(d)$ does not depend on~$y$.
It is then natural to ask for polynomial approximations where $\|p^{\ast}_d-y\|_{\infty}$ not only satisfies the above inequality, but also comes with an explicit upper bound satisfying a similar inequality, that is
\begin{equation} \label{eq:validated near minimax}
    \| p_d - y \|_{\infty} \leq B \leq \Lambda(d) \, \| p^{\ast}_d - y \|_{\infty},
    \qquad
    \text{$\Lambda$ independent of~$y$.}
\end{equation}
We thus leave open the following question.

\begin{question}
  Given a D-finite function~$y$ and a degree bound~$d$, what is the complexity of computing a pair $(p_d, B)$ with $\deg p_d \leq d$ satisfying~\eqref{eq:validated near minimax} for some~$\Lambda(d)$?
  For instance, can it be done in~$O(d) + \ln(\|p^{\ast}_d -y\|_{\infty}^{-1})$ arithmetic operations when $y$~is fixed?
\end{question}

Another subject for future work is the following.
In the timespan since we prepared the first draft of this work, an article by Olver and Townsend~\cite{OlverTownsend2013} has appeared that studies a similar question---how to obtain polynomial approximations of solutions of linear ODEs on the Chebyshev basis ``in linear time''---from a Numerical Analysis perspective.
On first sight at least, the motivations, language, and techniques look quite different from ours, and there appears to be little overlap between the actual results.
Yet the methods have common ingredients.
Roughly speaking, our algorithm may also be viewed as a \emph{coefficient spectral method} in the terminology of Olver and Townsend.
Their method is more general in the sense that it can deal with non-polynomial coefficients, which also means that they do not directly exploit the Chebyshev recurrence.
Instead, the computation of the approximation polynomials (for which we use a block Miller algorithm) boils down to the fast and numerically stable solution of a linear system similar to~\eqref{eq:exact linear constraints}.
There is no validation of the solution.
It is intriguing to understand these links in detail and determine if the best features of the two methods can somehow be combined.

Beyond non-polynomial coefficients, an interesting research direction concerns the case of nonlinear ODEs.
We may expect algorithms of a different kind (probably based on Newton's method instead of recurrences) for the computation of polynomial approximations, but some of the ideas used in the present article may still apply.
And, closer to what we do here, it is natural to ask for a generalization to other families of orthogonal polynomials, starting with the rest of the class of Gegenbauer polynomials.

\section*{Acknowledgements}

We thank Alin Bostan, Nicolas Brisebarre, Élie de Panafieu, Bruno Salvy and
Anne Vaugon for useful discussions and/or comments on various drafts of this work, and Moulay Barkatou for pointing out Ramis'
method to us.

\bibliographystyle{abbrv}
\bibliography{unifapprox}

\end{document}